\documentclass{article}
\usepackage{amsmath}
\usepackage{latexsym}
\usepackage{graphicx}
\usepackage{color}
\usepackage{amssymb}
\usepackage{subfigure}
\usepackage{tabularx}
\usepackage[labelfont={bf,sf},tableposition=top]{caption}

\newtheorem{theorem}{Theorem}
\newtheorem{appointing}{Appointing}
\newtheorem{acknowledgement}{Acknowledgement}
\newtheorem{definition}{Definition}

\newtheorem{proposition}{Proposition}
\newtheorem{remark}{Remark}
\newtheorem{conjecture}{Conjecture}
\newtheorem{proof}{Proof}

\begin{document}
\title{Study of the numerical solutions for the Electrical Impedance Equation in the plane: A pseudoanalytic approach of the forward Dirichlet boundary value problem.}
\author{M. P. Ramirez T.$^{(1)}$, C. M. A. Robles G.$^{(1)(2)}$\\
and R. A. Hernandez-Becerril$^{(1)(2)}$\\
\small{$^{(1)}$Communications and Digital Signal Processing Group,}\\ 
\small{Faculty of Engineering, La Salle University,}\\
\small{B. Franklin 47, C.P. 06140, Mexico. marco.ramirez@lasallistas.org.mx} \\ 
\small{$^{(2)}$SEPI, ESIME Culhuacan, National Polytechnic Institute,}\\
\small{Av. Santa Ana No. 1000, C.P. 04430, Mexico.}}

\date{}
\maketitle
\begin{abstract}
Employing a limiting case of a conjecture for constructing piecewise separable-variables functions, the elements of the Pseudoanalytic Function Theory are used for numerically approaching solutions of the forward Dirichlet boundary value problem, corresponding to the Electrical Impedance Equation in the plane, when the electrical conductivity is an arbitrary non-vanishing function, fully defined within a bounded domain. The new method is studied considering a variety of examples when the bounded domain coincides with the unit circle, and it is also included a description of its behaviour in non-smooth domains, selecting special cases that do not require additional regularization techniques, for warranting the convergence of the approach at the non-smooth regions, when certain requirements are fulfilled.
\end{abstract}

\section{Introduction.}

The study of the Electrical Impedance Equation
\begin{equation}
\nabla\cdot\left(\sigma\nabla u\right)=0,
\label{int:00}
\end{equation}
where $\sigma$ represents the conductivity function and $u$ denotes the electric potential, is fundamental for the proper understanding of a wide variety of physical problems, among which we find the Electrical Impedance Tomography, correctly stated in mathematical form by A.P. Calderon \cite{calderon} in 1980. Indeed, most of the algorithms that approach solutions for this inverse problem, within a bounded domain $\Omega\left(\mathbb{R}^{2}\right)$, are based on iterative methods that examine solutions for the direct problem, and introduce a certain kind of variations in the conductivity $\sigma$, attempting to minimize the difference between the approached solution $u$, and the boundary condition $u_{\textbf{c}}\vert_{\Gamma}$(see \emph{e.g.}, the classic work \cite{webster}). Here $\Gamma$ denotes the boundary of the domain $\Omega$.

Yet, specifically talking about the Electrical Impedance Tomography in the plane, the mathematical complexity of (\ref{int:00}) has posed so strong challenges, that the problem is still considered ill posed. This is, when small variations of the conductivity $\sigma$ are introduced for minimizing the error, most of the classical numerical methods do not achieve to decrease such error under a certain limit.

At this point, it is convenient to empathize that the majority of these methods, are based upon variations of the Finite Element Method, one of the finest tools for solving partial differential equations in bounded domains $\Omega\left(\mathbb{R}^{2}\right)$. Therefore, the usage of techniques upcoming from completely different branches of the Applied Mathematics, could well show up new information about the behaviour of the solutions for the forward Dirichlet boundary value problem of (\ref{int:00}) in the plane, which eventually could allow to propose new techniques for analysing its inverse problem. If such becomes true, its common classification of ill posed problem could be reconsidered. 

In this sense, the discovering of the relation between the two-dimensional case of (\ref{int:00}) and the Vekua equation \cite{vekua}, independently achieved by V. Kravchenko in 2005 \cite{kra2005}, and K. Astala and L. P\"aiv\"arinta in 2006 \cite{astala}, possesses special relevance, since it opened a new path for analysing the solutions of (\ref{int:00}) in the plane, from the point of view of the so-called Taylor series in formal powers \cite{bers}. As a matter of fact, the Vekua equation had been deeply studied in a variety of interesting works, published almost five decades before its relation with (\ref{int:00}) would have been first noticed. Indeed, two of the most important and complete works about the Vekua equation, were published by L. Bers \cite{bers} in 1953, and I. Vekua \cite{vekua} itself, in 1962.

The number of publications that came after discovering the cited relation is long, and they are all interesting (see \emph{e.g.} \cite{cck} and \cite{ckr}). Yet, we shall remark that only some of them can be directly employed in Physics (see e.g. \cite{kpa}), since it is not always clear how to adapt the elements of the Modern Pseudoanalytic Function Theory to the experimental physical requirements.

The current pages intend to make a positive contribution in this direction. We will not reach the study of the Electrical Impedance Tomography problem, but we will provide the elements that allow the employment of the Modern Pseudoanalytic Function Theory, in a variety of cases that can be easily identified with physical experimental models.

More precisely, after providing the necessary elements of the pseudoanalytic functions, we expose the numerical method that allows the construction of the formal powers, which eventually, will provide a complete set of solutions for solving the forward Dirichlet boundary value problem of (\ref{int:00}), when $\Omega$ coincides with the unit circle. The performance of the numerical method is tested by employing one example of conductivity  that possesses a separable-variables form, and for which an exact solution is known. Indeed, that exact solution will be imposed as the boundary condition $u_{\textbf{c}}\vert_{\Gamma}$, so we can estimate the accuracy of the approaching by a measure $\mathcal{E}$, that coincides with the Lebesgue integral over the boundary $\Gamma$.

Subsequently, we use a Conjecture \cite{ioprrh} for constructing piecewise separable-variables conductivity functions, a requisite for studying the solutions of (\ref{int:00}) from the point of view of the Pseudoanalytic Function Theory \cite{kpa}. We examine the effectiveness of this idea solving once more the boundary value problem posed previously, but employing the piecewise conductivity in lieu of the original $\sigma$, making a comparison between the errors $\mathcal{E}$ reached in both cases.

Immediately after, we present a Proposition stating that any conductivity function $\sigma:\Omega\left(\mathbb{R}^{2}\right)\rightarrow\mathbb{R}$, can be considered the limiting case of a piecewise separable-variables function, at every point $(x,y)\in\Omega$. As a matter of fact, this Proposition is the vertebral column of this work, because it allows the study of the forward Dirichlet boundary value problem of (\ref{int:00}), for the cases when the exact representation of $\sigma$ is known, but does not posses a separable-variables form. Moreover, for the first time in the literature dedicated to the Applied Pseudoanalytic Function Theory, the numerical method will be employed for analysing conductivities upcoming from geometrical distributions, a very important fact for physical applications.

We close this work with a review of the behaviour of the method, when the domain $\Omega$ does not coincide with the unit circle, and possesses avoidable discontinuities in the derivative of the parametric curve describing $\Gamma$. To better analyse this case, we selected three kind of conductivities, two of them rising from geometrical distributions, that do not require any additional regularization technique, in order to warrant the convergence of the approached solutions in the corner points, when some certain conditions are fulfilled.

Notice that, even a formal comparison of this new technique with an adequate variation of the Finite Element Method is in order, on behalf of briefness we refer the reader to the results reported in \cite{ckr}, where an accurate likening of this type was performed. The reader will find that, nonetheless the contrast was made employing separable-variables conductivities, the results can be easily extended to the cases treated in the current pages.
 
\section{Preliminaries.}

Following \cite{bers}, let the complex-valued functions $F$ and $G$ satisfy the condition
\begin{equation}
\mbox{Im}\left(\overline{F}G\right)>0,
\label{pre:00}
\end{equation}
where $\overline{F}$ denotes the complex conjugate of $F$: $\overline{F}=\mbox{Re}F-i\mbox{Im}F$, and $i$ is the standard imaginary unit: $i^{2}=-1$. Thus, any complex-valued function $W$ can be expressed by means of the linear combination of $F$ and $G$:
\[
W=\phi F+\psi G,
\]
where $\phi$ and $\psi$ are purely real functions. Two complex functions that fulfil (\ref{pre:00}) shall be called an $(F,G)$-\emph{generating pair}. Bers \cite{bers} introduced the $(F,G)$-\emph{derivative} of the function $W$ according to the expression:
\begin{equation}
\partial_{(F,G)}W=\left(\partial_{z}\phi\right)F+\left(\partial_{z}\psi\right)G.
\label{pre:01}
\end{equation}
This derivative will exist if and only if
\begin{equation}
\left(\partial_{\overline{z}}\phi\right)F+\left(\partial_{\overline{z}}\psi\right)G=0,
\label{pre:02}
\end{equation}
where
\[
\partial_{z}=\partial_{x}-i\partial_{y},\ \ \ \partial_{\overline{z}}=\partial_{x}+i\partial_{y}.
\]

Notice that these operators are classically introduced with the factor $\frac{1}{2}$, but it will result more convenient to omit it in this work.

Introducing the functions
\begin{eqnarray}
A_{(F,G)}=\frac{\overline{F}\partial_{z}G-\overline{G}\partial_{z}F}{F\overline{G}-G\overline{F}},\ \ \ a_{(F,G)}=-\frac{\overline{F}\partial_{\overline{z}}G-\overline{G}\partial_{\overline{z}}F}{F\overline{G}-G\overline{F}}, \nonumber \\
B_{(F,G)}=\frac{F\partial_{z}G-G\partial_{z}F}{F\overline{G}-G\overline{F}},\ \ \ b_{(F,G)}=-\frac{G\partial_{\overline{z}}F-F\partial_{\overline{z}}G}{F\overline{G}-G\overline{F}};
\label{pre:03}
\end{eqnarray}
the expression of the $(F,G)$-derivative (\ref{pre:01}) will turn into
\begin{equation}
\partial_{(F,G)}W=\partial_{z}W-A_{(F,G)}W-B_{(F,G)}\overline{W},
\label{pre:04}
\end{equation}
and the condition (\ref{pre:02}) will be written as
\begin{equation}
\partial_{\overline{z}}W-a_{(F,G)}W-b_{(F,G)}\overline{W}=0.
\label{pre:05}
\end{equation}

The functions defined in (\ref{pre:03}) are called the \emph{characteristic coefficients} of the generating pair $(F,G)$, whereas the functions $W$, solutions of the equation (\ref{pre:05}), are named $(F,G)$-\emph{pseudoanalytic} functions. As a matter of fact, the equation (\ref{pre:05}) is known as the \emph{Vekua equation} \cite{vekua}, and it is the foundation of the present work in many senses.

The following statements were originally presented in \cite{bers} and \cite{kpa}. They have been adapted here for the purposes of this work.
\begin{theorem}
\label{th:00}
The elements of the generating pair $(F,G)$ are $(F,G)$-pseudoanalytic:
\[
\partial_{(F,G)}F=\partial_{(F,G)}G=0.
\]
\end{theorem}
\begin{remark}
Let $p$ be a non-vanishing function within a bounded domain $\Omega\left(\mathbb{R}^{2}\right)$. The functions
\begin{equation}
F_{0}=p, \ \ \ G_{0}=\frac{i}{p},
\label{pre:06}
\end{equation}
constitute a generating pair, whose characteristic coefficients are
\begin{eqnarray}
\begin{array}{c}
A_{\left(F_{0},G_{0}\right)}=a_{\left(F_{0},G_{0}\right)}=0,\\
B_{\left(F_{0},G_{0}\right)}=p^{-1}\partial_{z}p,\\ b_{\left(F_{0},G_{0}\right)}=p^{-1}\partial_{\overline{z}}p.
\end{array}
\label{pre:07}
\end{eqnarray} 
\end{remark}
\begin{definition}
Let $\left(F_{0},G_{0}\right)$ and $\left(F_{1},G_{1}\right)$ be two generating pairs of the form (\ref{pre:07}), and let their characteristic coefficients satisfy the relation
\begin{equation}
B_{\left(F_{1},G_{1}\right)}=-b_{\left(F_{0},G_{0}\right)}.
\label{pre:08}
\end{equation} 
Thus, the pair $\left(F_{1},G_{1}\right)$ will be called a successor of the pair $\left(F_{0},G_{0}\right)$, whereas $\left(F_{0},G_{0}\right)$ will be called a predecessor of $\left(F_{1},G_{1}\right)$.
\end{definition}
\begin{definition}
Let 
\begin{equation}
\left\lbrace \left(F_{m},G_{m}\right) \right\rbrace,\ m=0,\pm 1,\pm 2,...
\nonumber
\end{equation}
be a set of generating pairs, where every $\left(F_{m+1},G_{m+1}\right)$ is a successor of $\left(F_{m},G_{m}\right)$. Therefore, the set $\left\lbrace \left(F_{m},G_{m}\right)\right\rbrace$ will be called a generating sequence. Moreover, if there exist a number $c$ such that $\left(F_{m},G_{m}\right)=\left(F_{m+c},G_{m+c}\right)$ the generating sequence will be periodic, with period $c$.

Finally, if $\left(F,G\right)=\left(F_{0},G_{0}\right)$, we will say that the generating pair $\left(F,G\right)$ is embedded into the generating sequence $\left\lbrace \left(F_{m},G_{m}\right) \right\rbrace$.
\end{definition}

\begin{theorem}
\label{th:01}
Let $\left(F,G\right)$ be a generating pair of the form (\ref{pre:06}), and let $p$ be a separable-variables function:
\begin{equation}
p=p_{1}(x)p_{2}(y),
\nonumber
\end{equation}
where $x,y\in\mathbb{R}$. Thus $\left(F,G\right)$ will be embedded into a periodic generating sequence, with period $c=2$, such that
\[
F_{m}=\frac{p_{2}(y)}{p_{1}(x)},\ \ G_{m}=i\frac{p_{1}(x)}{p_{2}(y)};
\]
when the subindex $m$ is an even number, and  
\[
F_{m}=p_{1}(x)p_{2}(y),\ \ G_{m}=\frac{i}{p_{1}(x)p_{2}(y)};
\]
when $m$ is odd. 

Furthermore, if particularly $p_{1}(x)=1$, it is easy to verify that the generating sequence in which $\left(F,G\right)$ is embedded will be periodic, but with period $c=1$.
\end{theorem}

L. Bers also introduced the concept of the $\left(F_{0},G_{0}\right)$-integral of a complex-valued function $W$. We refer the reader to the specialized literature \cite{bers} and \cite{kpa} for a detailed description of the necessary conditions for its existence. In the current pages, every complex function contained into an $\left(F_{0},G_{0}\right)$-integral will be, by definition, integrable.

\begin{definition}
\label{definition_adjoin}
Let $\left(F_{0},G_{0}\right)$ be a generating pair of the form (\ref{pre:06}). Its adjoin generating pair $\left(F_{0}^{*},G_{0}^{*}\right)$ is defined according to the formulas
\begin{equation}
F_{0}^{*}=-iF_{0},\ \ G_{0}^{*}=-iG_{0}.
\nonumber
\end{equation}
\end{definition}
\begin{definition}
\label{BersIntegral}
The $\left(F_{0},G_{0}\right)$-integral of a complex-valued function $W$ (when it exists \cite{bers}) is defined as:
\begin{equation}
\int_{\eta}Wd_{\left(F_{0},G_{0}\right)}z=F_{0}\mbox{Re}\int_{\eta}G_{0}^{*}Wdz+G_{0}\mbox{Re}\int_{\eta}F_{0}^{*}Wdz,
\nonumber
\end{equation}
where $\eta$ is a rectifiable curve within a domain in the complex plane. Specifically, if we consider the $\left(F_{0},G_{0}\right)$-integral of the $\left(F_{0},G_{0}\right)$-derivative of $W$, we  will have that:
\begin{equation}
\int_{z_{0}}^{z}\partial_{\left(F_{0},G_{0}\right)}Wd_{\left(F_{0},G_{0}\right)}z=-\phi (z_{0})F(z)-\psi (z_{0})G(z)+W(z),
\label{pre:09}
\end{equation}
where $z=x+iy$, and $z_{0}$ is a fixed point in the complex plane. According to the Theorem \ref{th:00}, the $\left(F_{0},G_{0}\right)$-derivatives of $F_{0}$ and $G_{0}$ vanish identically, thus the expression (\ref{pre:09}) can be considered the $\left(F_{0},G_{0}\right)$-antiderivative of $\partial_{\left(F_{0},G_{0}\right)}W$.
\end{definition}

\subsection{Formal Powers.}

\begin{definition}
\label{def:00}
The formal power $Z_{m}^{(0)}\left(a_{0},z_{0};z\right)$ belonging to the generating pair $\left(F_{m},G_{m}\right)$, with formal degree $0$, complex constant coefficient $a_{0}$, center at $z_{0}$, and depending upon $z=x+iy$, is defined according to the expression:
\begin{equation}
Z_{m}^{(0)}\left(a_{0},z_{0};z\right)=\lambda F_{m}(z)+\mu G_{m}(z),
\label{def:00a}
\end{equation}
where $\lambda$ and $\mu$ are complex constants fulfilling the condition:
\begin{equation}
\lambda F_{m}(z_{0})+\mu G_{m}(z_{0})=a_{0}.
\nonumber
\end{equation}

The formal powers with higher degrees are approached according to the recursive formulas:
\begin{equation}
Z_{m}^{(n)}\left(a_{n},z_{0};z\right)=n\int_{z_{0}}^{z}Z_{m-1}^{(n-1)}\left(a_{n},z_{0};z\right)d_{\left(F_{m},G_{m}\right)}z,
\label{pre:10}
\end{equation}
where $n=1,2,3,...$. Notice the integral operators in the right-hand side of the last expression are $\left(F_{m},G_{m}\right)$-antiderivatives.
\end{definition}

\begin{theorem}
\label{propertiesFormalPowers}
The formal powers posses the following properties:
\begin{enumerate}
\item Every $Z_{m}^{(n)}\left(a_{n},z_{0};z\right),\ n=0,1,2,...$ is an $\left(F_{m},G_{m}\right)$-pseudoanalytic function.
\item Let $a_{n}=a'_{n}+ia''_{n}$, where $a'_{n},a''_{n}\in\mathbb{R}$. The following relation holds
\begin{equation}
Z_{m}^{(n)}\left(a_{n},z_{0};z\right)=a'_{n}Z_{m}^{(n)}\left(1,z_{0};\right)+a''_{n}Z_{m}^{(n)}\left(i,z_{0};z\right).
\end{equation}
\item Finally
\begin{equation}
\lim_{z\rightarrow z_{0}}Z_{m}^{(n)}\left(a_{n},z_{0};z\right)=a_{n}(z-z_{0})^{n}.
\end{equation}
\end{enumerate}
\end{theorem}
\begin{theorem}
Let $W$ be an $\left(F_{m},G_{m}\right)$-pseudoanalytic function. Then it can be expressed in terms of the so-called Taylor series in formal powers:
\begin{equation}
W=\sum_{n=0}^{\infty}Z_{m}^{(n)}\left(a_{n},z_{0};z\right).
\label{pre:11}
\end{equation}
Furthermore, since any $\left(F_{m},G_{m}\right)$-pseudoanalytic function $W$ accepts this expansion, (\ref{pre:11}) is an analytical representation of the general solution for the Vekua equation (\ref{pre:07}).  
\end{theorem}

\subsection{The two-dimensional Electrical Impedance Equation.}

Let us consider the equation (\ref{int:00}) in the plane: 
\begin{equation}
\nabla\cdot\left(\sigma\nabla u\right)=0.
\nonumber
\end{equation}
As it has been shown in several previous works (see \emph{e.g.} \cite{kpa} and \cite{oct}), if $\sigma$ can be expressed by means of a separable-variables function:
\begin{equation}
\sigma(x,y)=\sigma_{1}(x)\sigma_{2}(y),
\nonumber
\end{equation}
by introducing the notations
\begin{equation}
\begin{array}{c}
W=\sqrt{\sigma}\partial_{x}u-i\sqrt{\sigma}\partial_{y}u,\\
p=\left(\sqrt{\sigma_{1}}\right)^{-1}\sqrt{\sigma_{2}};
\end{array}
\label{eie:00}
\end{equation}
the equation (\ref{int:00}) will turn into the Vekua equation
\begin{equation}
\partial_{\overline{z}}W-\frac{\partial_{\overline{z}}p}{p}\overline{W}=0,
\label{eie:01}
\end{equation}
for which the functions
\begin{equation}
F_{0}=p,\ \ G_{0}=\frac{i}{p},
\label{eie:02}
\end{equation}
conform a generating pair. 

From (\ref{eie:00}), and according to the Theorem \ref{th:01}, it is possible to verify that this pair is embedded into a generating sequence, with period $c=2$, for $p$ is a separable-variables function.

\subsection{A complete set for the Dirichlet boundary value problem of the two-dimensional Electrical Impedance Equation.}

An explicit generating sequence allows the construction of the formal powers (\ref{pre:10}), and in consequence, the approaching of the general solution of (\ref{int:00}), according to the relations (\ref{eie:00}).

Indeed, a very important relation between the solutions of (\ref{int:00}) and of (\ref{eie:01}) was elegantly posed in \cite{cck}, and this relation will play a central role in the present work.

\begin{theorem}
\label{TheoremComplete}
\cite{cck} Let us consider the set of formal powers
\begin{equation}
\left\lbrace Z_{0}^{(n)}\left(1,0;z\right),\ Z_{0}^{(n)}\left(i,0;z\right) \right\rbrace_{n=0}^{\infty},
\nonumber
\end{equation}
corresponding to the generating pair (\ref{eie:02}), and let $\Omega\left(\mathbb{R}^2\right)$ be a bounded domain, with boundary $\Gamma$, such that $0\in\Omega$ but $0\notin\Gamma$. Then the set of functions defined on $\Gamma$:
\begin{equation}
\left\lbrace \mbox{Re}Z_{0}^{(n)}\left(1,0;z\right)\vert_{\Gamma},\ \mbox{Re}Z_{0}^{(n)}\left(i,0;z\right)\vert_{\Gamma} \right\rbrace_{n=0}^{\infty},
\label{eie:03}
\end{equation}
conforms a complete system for approaching solutions of the forward Dirichlet boundary value problem of (\ref{int:00}).
\end{theorem}

In other words, according to the second property of the Theorem \ref{propertiesFormalPowers}, when a separable-variables conductivity function $\sigma$ is given within a bounded domain $\Omega$, and a boundary condition $u_{\textbf{c}}\vert_{\Gamma}$ is imposed for the solution of (\ref{int:00}), it will be always possible to construct a finite set of functions, subset of (\ref{eie:03}), such that
\begin{equation}
\oint\left( u_{\textbf{c}}\vert_{\Gamma}-\sum_{n=0}^{N}a_{n}'\mbox{Re}Z_{0}^{(n)}\left(1,0;z\right)\vert_{\Gamma}+a_{n}''\ \mbox{Re}Z_{0}^{(n)}\left(i,0;z\right)\vert_{\Gamma}\right)^{2}dl<\varepsilon,
\label{eie:04}
\end{equation}
where $\varepsilon>0$ and $l\in\Gamma$.

\subsection{Construction of a piecewise separable-variables function.}
\label{piecewise}

Approaching solutions of the forward Dirichlet boundary value problem for (\ref{int:00}), by employing formal powers, has proven its effectiveness in a variety of works (see \emph{e.g.} \cite{cck} and \cite{ckr}). Yet, how to apply those methodologies when the conductivity $\sigma$ is not represented as a separable-variables function, remains an open question. 

A possibility for studying these cases could be to introduce an interpolating method that, given a set of conductivity values defined into a bounded domain in the plane, it can be able to approach a piecewise separable-variables function. One of the first proposals in this direction was posed in \cite{oct}, and the next paragraphs will show that, even it is a basic idea, it can well serve to our main objectives.

Consider a bounded domain $\Omega$ (a unitary disk, for instance), and divide it into a finite number of subsections, taking care that the center $z_{0}$ of the formal powers (see the Definition \ref{def:00}) does not reside onto the boundary of two or more subsections. On behalf of simplicity, let us make the division by employing a finite set of parallel lines to the $y$-axis, equidistantly located one to each other, and let us fix $z_{0}=0$.

A simplified illustration of this steps is plotted in Figure \ref{fig:eie00a}. For this example, we consider $K=3$ subsections, delimited by the set of $K+1$ $y$-axis parallel lines 
\begin{equation}
\left\lbrace X_{(0)},X_{(1)},X_{(2)},X_{(3)}\right\rbrace.
\nonumber
\end{equation}

\begin{figure}
\centering
\subfigure[A circular sectioned domain.]{
\includegraphics[scale=0.350]{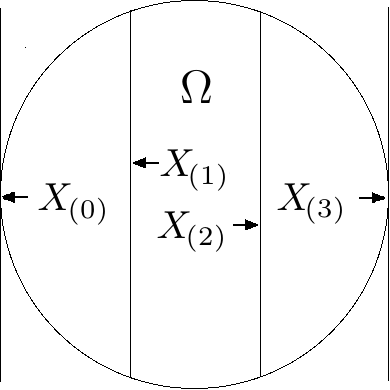}
\label{fig:eie00a}
}
\subfigure[Points on the lines $\varphi_{k}$, on which the interpolating functions $f_{(k)}(y)$ will be constructed.]{
\includegraphics[scale=0.350]{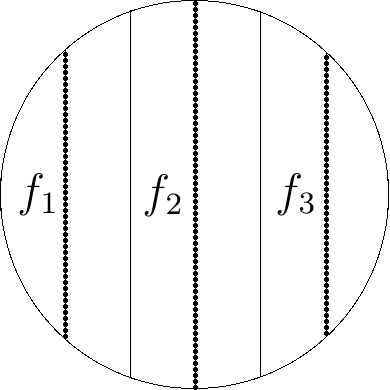}
\label{fig:eie00b}
}
\caption{Simplified illustrations of the steps for approaching a piecewise separable-variables conductivity function.}
\end{figure}

Supposing that the values of the electrical conductivity are defined at every point inside the domain $\Omega$, let us trace a straight line within every subsection, watching out that it does not intersect the $y$-parallel lines $X_{(k)}$ and $X_{(k+1)}$, that delimit its corresponding subsection. Indeed, such lines can simply be parallel to $X_{(k)}$ and $X_{(k+1)}$. We will introduce them as:
\begin{equation}
\varphi_{(k)}=\frac{x_{k+1}-x_{k}}{2},
\label{chi:inter}
\end{equation}
where $x_{k}$ is the common $x$-coordinate of all points belonging to $X_{(k)}$. In the same sense, by $\chi_{k}$ we will represent the common $x$-coordinate of the lines $\varphi_{(k)}$.

The next step is to collect a finite set of values $J$ over each line $\varphi_{(k)}$. For every crossing line, the quantity $J$ of collected values must be large enough to warrant that the interpolating functions $\left\lbrace f_{(k)}\right\rbrace$ (in our example interpolated by piecewise-defined polynomials of degree $1$) will adequately approach the remaining conductivity values defined over the line $\varphi_{(k)}$. This step of the procedure is illustrated in Figure \ref{fig:eie00b}. 

Let us assume that the conductivity inside every subsection is represented according to the expression:
\begin{equation}
\frac{x+A_{(k)}}{\chi_{k}+A_{(k)}}\cdot f_{(k)}(y),
\label{abv00}
\end{equation}
where $\chi_{k}$ denotes the common $x$-coordinate of the function $\varphi_{(k)}$, introduced in (\ref{chi:inter}), $f_{(k)}(y)$ is the interpolating function approached with the $J$ values of the conductivity, recollected over the $\varphi_{(k)}$ line, and $A_{(k)}$ is a positive real constant such that $x+A_{(k)}\neq 0$ within every subsection.

From this point of view, and supposing that we have $K$ subsections, the conductivity $\sigma$ inside the bounded domain $\Omega$, can be approached by means of the piecewise function:

\begin{equation}
   \sigma_{\textbf{pw}} (x,y) = \left\lbrace
     \begin{array}{l c l}
       \frac{x+A_{(1)}}{\chi_{1}+A_{(1)}}\cdot f_{(1)}(y) &:& x \in \left[x_{0},x_{1}\right);\\
       \frac{x+A_{(2)}}{\chi_{2}+A_{(2)}}\cdot f_{(2)}(y) &:& x \in \left[x_{2},x_{3}\right);\\
       \cdots & & \\
       \frac{x+A_{(K)}}{\chi_{K}+A_{(K)}}\cdot f_{(K)}(y) &:& x \in \left[x_{K-1},x_{K}\right].
     \end{array}
   \right.   
\label{abv01}
\end{equation}
Here, as appointed in Figure \ref{fig:eie00a}, $x_{0}$ represents the common $x$-coordinate of the first $y$-axis parallel line $X_{(0)}$ employed for subdividing $\Omega$, whereas $x_{K}$ represents the common coordinate of $X_{(K)}$. It is evident that the piecewise function (\ref{abv01}) is separable-variables.

Then, according to the Theorem \ref{th:01}, and the relations (\ref{eie:00}), it follows that

\begin{equation}
   F_{0} = \left\lbrace
     \begin{array}{l c l}
       \left(\frac{\chi_{1}+A_{1}}{x+A_{(1)}}\cdot f_{(1)}(y)\right)^{\frac{1}{2}} & : & x \in [x_{0},x_{1});\\
       \left(\frac{\chi_{2}+K_{2}}{x+K_{(2)}}\cdot f_{(2)}(y)\right)^{\frac{1}{2}} & : & x \in [x_{1},x_{2});\\
       \cdots & &\\
       \left(\frac{\chi_{K}+A_{(K)}}{x+A_{(K)}}\cdot f_{(K)}(y)\right)^{\frac{1}{2}} & : & x \in [x_{K-1},x_{K}];
     \end{array}
   \right. \nonumber
\end{equation}
whereas
\begin{equation}
   G_{0} = \left\lbrace
     \begin{array}{l c l}
       i\left(\frac{x+A_{(1)}}{\chi_{1}+A_{(1)}}\cdot \frac{1}{f_{(1)}(y)} \right)^{\frac{1}{2}} & : & x \in [x_{0},x_{1});\\
       i\left(\frac{x+A_{(2)}}{\chi_{(2)}+A_{(2)}}\cdot \frac{1}{f_{(2)}(y)}\right)^{\frac{1}{2}} & : & x \in [x_{1},x_{2});\\
       \cdots & \\
       i\left(\frac{x+A_{(K)}}{\chi_{K}+A_{(K)}}\cdot \frac{1}{f_{(K)}(y)}\right)^{\frac{1}{2}} & : & x \in [x_{K-1},x_{K}].
     \end{array}
   \right. \nonumber  
\end{equation}

For the generating pair $(F_{1},G_{1})$ we will simply have that
\begin{equation}
F_{1}=\sqrt{\sigma_{\textbf{pw}}},\ \ \ G_{1}=i\left(\sqrt{\sigma_{\textbf{pw}}}\right)^{-1};
\nonumber
\end{equation}

These are the generating pairs $(F_{0},G_{0})$ and $(F_{1},G_{1})$ that we will employ for the numerical approach the formal powers.

\section{Numerical solutions of the forward Dirichlet boundary value problem.}

In order to evaluate the effectiveness of the piecewise function $\sigma_{\textbf{pw}}$, introduced in (\ref{abv01}), let us approach the solution of the forward Dirichlet boundary value problem corresponding to (\ref{int:00}), at the perimeter of the unit circle, imposing an exact solution as the boundary condition $u_{\textbf{c}}\vert_{\Gamma}$.

\begin{proposition}
\label{prop02}
Let
\begin{equation}
\sigma=\left(\frac{1}{x^2+0.1}\right)\left(\frac{1}{y^2+0.1}\right).
\label{bs09}
\end{equation} 
Then the function
\begin{equation}
u=\frac{x^3+y^3}{3}+0.1\left(x+y\right),
\label{bs10}
\end{equation}
will be a particular solution of (\ref{int:00}).
\end{proposition}

\subsection{Numerical approach of the formal powers.}
\label{highacc}

We will study a numerical method for approaching elements of the set (\ref{eie:03}). A more detailed description of this method, including a variety of special examples, can be found in \cite{bucio}. In the next paragraphs, we will focus our attention in the construction of the subset of formal powers
\begin{equation}
\left\lbrace Z_{0}^{(n)}(1,0;z),Z_{1}^{(n)}(1,0;z)\right\rbrace_{n=0}^{N},\nonumber
\end{equation}
because not any significant alteration is needed when approaching the formal powers with coefficient $a_{n}=i$.

Taking into account that the integral operators introduced in (\ref{pre:10}) are path-independent \cite{bers}\cite{kpa}, let us consider a radius $R$ of the unit circle with center at $z_{0}=0$, as the rectifiable curve $\eta$ described in (\ref{BersIntegral}).

We shall consider $P+1$ points equidistantly distributed on $R$, being the first $r[0]=0$ and the last $r[P]=1$:
\begin{equation}
\left\lbrace r[p]=\frac{p}{P}\right\rbrace_{p=0}^{P}.
\label{bs:10aa}
\end{equation}
Thus we can construct a set of coordinates according to the formulas:
\begin{equation}
\begin{array}{c c c}
x[p]&=&r[p]\cos\theta_{q},\\
y[p]&=&r[p]\sin\theta_{q};
\end{array}
\label{bs:10a}
\end{equation}
where $\theta_{q}$ is the angle corresponding to $R$.

According to (\ref{eie:00}), the data (\ref{bs:10a}) will be used to obtain the sets of values 
\begin{equation}
\begin{array}{c}
F_{0}(z[p])=\left(y[p]^2+0.1\right)^{-\frac{1}{2}}\left(x[p]^2+0.1\right)^{\frac{1}{2}},\\
F_{1}(z[p])=\left(y[p]^2+0.1\right)^{-\frac{1}{2}}\left(x[p]^2+0.1\right)^{-\frac{1}{2}};
\end{array}
\nonumber
\end{equation}
where $z[p]=x[p]+iy[p]$. Their associated functions $G_{0}(z[p])$ and $G_{1}(z[p])$ will be constructed according to (\ref{pre:06}), whereas the adjoin pairs $(F_{0}^{*}(z[p]),G_{0}^{*}(z[p]))$ and $(F_{1}^{*}(z[p]),G_{1}^{*}(z[p]))$ will have the form introduced in the Definition \ref{definition_adjoin}.

From (\ref{def:00a}), it immediately follows that
\begin{equation}
\begin{array}{c c c}
Z_{0}^{(0)}(1,0;z[p])&=&F_{0}(z[p]),\\
Z_{1}^{(0)}(1,0;z[p])&=&F_{1}(z[p]);
\end{array}
\label{Z:00}
\end{equation}
Hereafter, each formal power with $n>0$ will be always approached considering $P+1$ equidistant points within the interval $[0,1]$.

Consider the formal powers $Z^{(n)}_{0}(1,0;z[p])$. To approach their values at the points $z[p]=x[p]+iy[p]$, we shall employ a variation of the trapezoidal integration method:
\begin{equation}
\begin{array}{c}
Z_{0}^{(n)}(1,0;z[p])=\\
\delta F_{0}(z[p])\cdot\mbox{Re}\sum\limits_{s=0}^{p-1}\left(Z_{1}^{(n-1)}(1,0;z[s+1])\cdot G_{0}^{*}\left(z[s+1]\right)\right)dz[s]+\\
+\delta F_{0}(z[p])\mbox{Re}\sum\limits_{s=0}^{p}\left(Z_{1}^{(n-1)}(1,0;z[s])\cdot G_{0}^{*}\left( z[s]\right)\right)
dz[s]+\\
+\delta G_{0}(z[p])\cdot\mbox{Re}\sum\limits_{s=0}^{p-1}\left(Z_{1}^{(n-1)}(1,0;z[s+1])\cdot F_{0}^{*}\left(z[s+1]\right)\right)dz[s]+\\
+\delta G_{0}(z[p])\mbox{Re}\sum\limits_{s=0}^{p}\left(Z_{1}^{(n-1)}(1,0;z[s])\cdot F_{0}^{*}\left( z[s]\right)\right)dz[s];
\end{array}
\label{high}
\end{equation}
where
\begin{equation}
dz[s]=\left(z[s+1]-z[s]\right),
\nonumber
\end{equation}
and $\delta$ is a real constant factor, heuristically selected, that contributes to the numerical stability of the method. In general, every experiment requires an individual set of trials in order to estimate an adequate value of $\delta$. Particularly, for the examples further considered in this work, $\delta=9$ provided de best results.  

It is also important to remark that once we have adopted the expression (\ref{high}) for approaching the formal powers, we implicitly employ a piecewise interpolating polynomial function of degree $1$, to relate every value $Z_{0}^{(n)}(1,0;z[p])$, for $p=0,1,...,P$; and $n=0,1,...,N$.

Notice also that, according to the third property of the Theorem \ref{propertiesFormalPowers}, for $\forall n>0$, we have that
\begin{equation}
Z_{0}^{(n)}(1,0;z[0])\equiv 0.\nonumber
\end{equation}

Iterating the last procedure we can approach $N+1$ formal powers that will conform the set
\begin{equation}
\left\lbrace \mbox{Re}Z_{0}^{(n)}(1,0;z)\right\rbrace_{n=0}^{N};
\nonumber
\end{equation}
and after making minor modifications, we shall also approach the set of $N$ formal powers
\begin{equation}
\left\lbrace \mbox{Re}Z_{0}^{(n)}(i,0;z)\right\rbrace_{n=1}^{N};
\nonumber
\end{equation}
pointing out that, according to the Theorem \ref{pre:06}, $\mbox{Re}Z_{0}^{(0)}(i,0;z_{\tau})\equiv 0$.

Performing the full procedure for a wide enough quantity $Q$ of radii $R$, each one at some angle $\theta_{q}$:
\begin{equation}
\left\lbrace\theta_{q}=q\cdot\frac{2\pi}{Q}\right\rbrace_{q=0}^{Q-1},
\label{totalangles}
\end{equation}
we will be able to approach $2N+1$ discrete elements of the set (\ref{eie:03}), introduced in the Theorem \ref{TheoremComplete}:

\begin{equation}
\left\lbrace \mbox{Re}Z^{(n)}_{0}(1,0,z)\vert_{\Gamma}\right\rbrace_{n=0}^{N}\cup\left\lbrace \mbox{Re}Z^{(n)}_{0}(i,0,z)\vert_{\Gamma}\right\rbrace_{n=1}^{N}.
\label{numerical:set}
\end{equation}

Finally, executing a classical Gram-Schmidt orthonormalizing process over such elements, and applying a standard cubic splines interpolating method, we will obtain a set of $2N+1$ orthonormal piecewise continuous functions, defined at every point of the boundary $\Gamma$:  
\begin{equation}
\left\lbrace u_{n}(l):l\in\Gamma\right\rbrace_{n=0}^{2N}.
\label{ortho:set}
\end{equation}

\subsection{Obtention of the constant coefficients for approaching the boundary condition.}
\label{coefficients:alpha}

Once we have build the set (\ref{ortho:set}), there exist a variety of techniques for approaching the coefficients $\{\alpha_{n}\}_{n=0}^{2N}$, that will accomplish the task of asymptotically attain the boundary condition $u_{\textbf{c}}\vert_{\Gamma}$ (see e.g. \cite{ckr}). Here, we will allocate as many points at the boundary as the number of functions contained in (\ref{ortho:set}).

More precisely, we will trace $2N+1$ radii from the center $z_0=0$, at the angles

\begin{equation}
\left\lbrace \omega_{n}=n\cdot\frac{2\pi}{2N+1}\right\rbrace_{n=0}^{2N},
\label{angle:set}
\end{equation}
and we will evaluate the boundary condition $u_{\textbf{c}}\vert_{\Gamma}$ on the points $(x_{n},y_{n})$, where the radii intersect the boundary $\Gamma$, obtaining a set of values $\{\gamma_{n}\}_{n=0}^{2N}$. 

We will assume that the number $N$ of formal powers is big enough to warrant that, given a number $\varepsilon>0$, there exist a set of real coefficients $\{\alpha_{n}\}_{n=0}^{2N}$ such that

\begin{equation}
\left(\int_{\Gamma}\left(\sum\limits_{n=0}^{2N}\alpha_{n}u_{n}(l)-u_{\textbf{c}}\vert_{\Gamma}\right)^{2}dl\right)^{\frac{1}{2}}<\varepsilon,\ l\in\Gamma.
\nonumber
\end{equation}

Thus, the following matrix relation must hold
\begin{equation}
\mathbb{U}\overrightarrow{\alpha}=\overrightarrow{\gamma},
\nonumber
\end{equation}
where
\begin{equation}
\mathbb{U}=\left[
\begin{array}{c c c c}
u_{0}(x_{0},y_{0}) & u_{1}(x_{0},y_{0}) & \cdots & u_{2N}(x_{0},y_{0})\\
u_{0}(x_{1},y_{1}) & u_{1}(x_{1},y_{1}) & \cdots & u_{2N}(x_{1},y_{1})\\
u_{0}(x_{2},y_{2}) & u_{1}(x_{2},y_{2}) & \cdots & u_{2N}(x_{2},y_{2})\\
\vdots & \vdots & \ddots & \vdots\\
u_{0}(x_{2N},y_{2N}) & u_{1}(x_{2N},y_{2N}) & \cdots & u_{2N}(x_{2N},y_{2N})
\end{array}
\right],
\label{matrix}
\end{equation}
\begin{equation}
\overrightarrow{\alpha}=\left[\alpha_{0};\alpha_{1};\alpha_{2};\cdots ;\alpha_{2N}\right];
\nonumber
\end{equation}
and
\begin{equation}
\overrightarrow{\gamma}=\left[\gamma_{0};\gamma_{1};\gamma_{2};\cdots ;\gamma_{2N}\right].
\nonumber
\end{equation}

Given the orthonormality of the set (\ref{ortho:set}), there will exist an inverse matrix $\mathbb{U}^ {-1}$ for (\ref{matrix}) such that

\begin{equation}
\overrightarrow{\alpha}=\mathbb{U}^ {-1}\overrightarrow{\gamma}.
\nonumber
\end{equation}

\subsection{Comparison between the solutions obtained employing the original conductivity, and the piecewise approached conductivity.}
\label{comparison}

We will perform the method described in the Subsection \ref{highacc}, considering the domain $\Omega$ as the unit circle, and a total error $\mathcal{E}$ of the form:
\begin{equation}
\mathcal{E}=\left(\int_{\Gamma}\left(\sum\limits_{n=0}^{2N}\alpha_{n}u_{n}(l)-u_{\textbf{c}}\vert_{\Gamma}\right)^{2}dl\right)^{\frac{1}{2}}.
\label{totalerror}
\end{equation}
Beside, we shall fix the following parameters:
\begin{enumerate}
\item Maximum number of formal powers $N=30$.
\item Total number of radii $Q=1000$, located at the set of angles defined in (\ref{totalangles}). 
\item Total number of points per radius, where the formal powers are defined, $P=1000$.
\end{enumerate}

Therefore, we will obtain a set of $61$ orthonormal functions defined at the boundary $\Gamma$. 

Valuing the boundary condition (\ref{bs10}) at the intersecting points of the $Q$ radii $R$ with $\Gamma$, and employing the conductivity $\sigma$ defined in (\ref{bs09}), the procedure detailed in the Subsection \ref{highacc}, reached a total error:
\begin{equation}
\mathcal{E}_{1}= 3.6786\times 10^{-9}.
\label{analyticerror}
\end{equation}

The same procedure was executed considering the piecewise separable-variables conductivity function $\sigma_{\textbf{pw}}(x,y)$, introduced in (\ref{abv01}), with the same parameters indicated in Subsection \ref{piecewise}, and assuming also that:
\begin{enumerate}
\item The number of subsections $K=1000$; and
\item The constants
\begin{equation}
A_{k}=60,\ k=1,2,...,1000.
\nonumber
\end{equation}
\end{enumerate}  
The resultant total error was:
\begin{equation}
\mathcal{E}=7.7263\times 10^{-7}.
\nonumber
\end{equation}

The magnitude of this total error indicates that $\sigma_{\textbf{pw}}(x,y)$, introduced in (\ref{abv01}), can effectively approach separable-variables functions, and in consequence, we can use it to apply the elements of the Pseudoanalytic Function Theory for analysing the forward Dirichlet boundary value problem of (\ref{int:00}), in the cases when the conductivities $\sigma$ do not originally posses a separable-variables form. 

\begin{conjecture}
\label{con:00}
Let $\sigma$ be a conductivity function, defined within a bounded domain $\Omega \left(\mathbb{R}^{2}\right)$, such that:
\begin{equation}
\sigma:\Omega\left(\mathbb{R}^{2} \right)\rightarrow\mathbb{R};
\nonumber
\end{equation} and let $\Gamma$ be the boundary of such domain. Then it is possible to approach $\sigma$ by means of a piecewise separable-variables function of the form (\ref{abv01}):
\begin{equation}
   \sigma_{\textbf{pw}} (x,y) = \left\lbrace
     \begin{array}{l c l}
       \frac{x+A_{(1)}}{\chi_{1}+A_{(1)}}\cdot f_{(1)}(y) &:& x \in [x_{0},x_{1});\\
       \frac{x+A_{(2)}}{\chi_{2}+A_{(2)}}\cdot f_{(2)}(y) &:& x \in [x_{1},x_{2});\\
       \cdots & & \\
       \frac{x+A_{(K)}}{\chi_{K}+A_{(K)}}\cdot f_{(K)}(y) &:& x \in [x_{K-1},x_{K}];
     \end{array}
   \right.
   \nonumber   
\end{equation}
which can be employed for constructing a finite set of formal powers of the form (\ref{eie:03}), in order to approach solutions of the forward Dirichlet boundary value problem corresponding to the equation (\ref{int:00}), when a boundary condition $u_{\textbf{c}}\vert_{\Gamma}$ is imposed.
\end{conjecture}

The postulate remains a conjecture because not any formal extension of the proof posed in \cite{cck}, about the completeness of the set (\ref{eie:03}), is known for the case of piecewise separable-variables functions within bounded domains.

\subsection{A special case of piecewise separable-variables conductivity functions.}

It is natural to inquire about the behaviour of the interpolating method described in the Subsection \ref{piecewise}, when changing the values of the constant parameters $K$, $A_{(k)}$ or $J$. In this direction, the work \cite{ioprrh} provides some information about such specific questions. Nevertheless, the validity of the Conjecture \ref{con:00} allows us to study one property that will significantly simplify our analysis.

\begin{proposition}
\label{pro:00}
Let $\Omega\left(\mathbb{R}^2\right)$ be a bounded domain, and let us denote by $\Gamma$ its boundary. Every conductivity function
\begin{equation}
\sigma:\Omega\left(\mathbb{R}^{2}\right)\rightarrow\mathbb{R};
\nonumber
\end{equation}
can be considered the limiting case of a piecewise separable-variables conductivity function $\sigma_{\textbf{pw}}$ of the form (\ref{abv01}), at every point $(x,y)\in\Omega$, when the number $K$ of subdomains introduced in Subsection \ref{piecewise}, and the number $J$ of collected values over $\varphi_{(k)}$, defined in (\ref{chi:inter}), tends to infinity:

\begin{equation}
\sigma(x,y)=\lim_{K,J\rightarrow\infty}\sigma_{\textbf{pw}}(x,y).
\label{eie:06}
\end{equation}
Moreover, since
\begin{equation}
\lim_{K\rightarrow\infty}\frac{x+A_{(k)}}{\chi_{k}+A_{(k)}}=1,
\nonumber
\end{equation}
from the Theorem \ref{th:01}, it follows that the generating sequence of this limiting case will be periodic, with period $c=1$. 
\end{proposition}

\begin{proof}
Let $x_{0}$ be the minimum of the subset of $x$-coordinates corresponding to the points $(x,y)\in\Omega$, and let $x_{K}$ be the maximum. We can divide $\Omega$ into $K$ subdomains $\left\lbrace\Omega_{k}\right\rbrace_{k=1}^{K}$, by employing a set of equidistant $y$-axis parallel lines (see Figure \ref{fig:eie00a}):
\begin{equation}
\left\lbrace X_{(k-1)}=x_{0}+\frac{(k-1)(x_{K}-x_{0})}{K} \right\rbrace_{k=1}^{K+1}.
\label{X-lines}
\end{equation}
Thus, the subdomains will be defined in the following form:
\begin{equation}
\left\lbrace \Omega_{k}\vert x,y\in\mathcal{L}_{k}\left(X_{(k-1)},X_{(k)}\right)\cap\Omega\right\rbrace_{k=1}^{K},
\nonumber
\end{equation}
where $\mathcal{L}_{k}\left(X_{(k-1)},X_{(k)}\right)$, represent the set of points $(x,y)$ within the subsection of the plane bounded by the pair of $y$-axis parallel lines $X_{(k-1)}$ and $X_{(k)}$. 

Let us consider the set of common $x$-coordinates belonging to the lines $\varphi_{(k)}$ introduced in (\ref{chi:inter}):
\begin{equation}
\left\lbrace \chi_{k}\right\rbrace_{k=1}^{K}.
\nonumber
\end{equation}
From (\ref{X-lines}), we have that
\begin{equation}
\lim_{K\rightarrow\infty}\vert x_{k}-x_{k-1}\vert=0;\ k=1,...,K;
\nonumber
\end{equation}
where $x_{k}$ are the common $x$-coordinate of the lines $X_{(k)}$, and in consequence we will have that
\begin{equation}
\forall x\in\Omega_{k} \ :\ x\rightarrow\chi_{k}; \\
\nonumber
\end{equation}
It immediately follows
\begin{equation}
\lim_{K\rightarrow\infty}\frac{x+A_{(k)}}{\chi_{k}+A_{(k)}}=1,
\label{limitproof}
\end{equation}
thus every subdomain $\Omega_{k}$ will be conformed by the points
\begin{equation}
\left\lbrace \Omega_{k} \ \vert \ (x,y)\in \varphi_{(k)}\cap\Omega\right\rbrace.
\nonumber
\end{equation} 

Furthermore, since the number $J$ of conductivity values, obtained by evaluating $\sigma$ along the line $\varphi_{(k)}$, also tends to infinite, not any interpolation method will be required for approaching $f_{k}(y)$. It will simply coincide with the values of $\sigma$ evaluated at $(\chi_{k},y)$.

Finally, it follows from (\ref{limitproof}) that
\[
\lim_{K\rightarrow\infty}\frac{x+A_{k}}{\chi_{k}+A_{k}}\cdot f_{k}(y)=f_{k}(y).
\]
Thus, according to the Theorem \ref{th:01}, the generating sequence for numerically approaching the subset of formal powers:
\begin{equation}
\left\lbrace Z_{0}^{(n)}(1,0;z),Z_{0}^{(n)}(i,0;z)\right\rbrace_{n=0}^{N}
\end{equation}
will be periodic, with period $c=1$.
\end{proof}

The last proposition indicates that the full procedure described in the Subsection \ref{highacc}, can be performed considering:
\begin{equation}
\begin{array}{c}
F_{0}=\sqrt{\sigma},\ G_{0}=i\left(\sqrt{\sigma}\right)^{-1},\\
(F_{0},G_{0})=(F_{1},G_{1}).
\end{array}
\nonumber
\end{equation}
Then, using the same parameters $N$, $Q$ and $P$, shown at the beginning of the Section \ref{comparison}, and after performing the full numerical procedure, the total error was
\begin{equation}
\mathcal{E}_{2}=4.2458\times 10^{-9},
\label{errorimp}
\end{equation}
which is indeed bigger than the one obtained in (\ref{analyticerror}), where all operations were performed strictly following the postulates of the Pseudoanalytic Function Theory \cite{bers}. Nevertheless, the error (\ref{errorimp}) is small enough for considering we have obtained an acceptable approach of the boundary condition.

\begin{remark}
Executing the same logical steps described in Subsection \ref{piecewise}, but dividing the domain $\Omega$ by a set of $x$-axis parallel lines
\begin{equation}
\left\lbrace Y_{(0)}, Y_{(1)},..., Y_{(K)}\right\rbrace,
\nonumber
\end{equation}
we can approach a piecewise separable-variables conductivity function of the form
\begin{equation}
   \sigma_{\textbf{pw}} (x,y) = \left\lbrace
     \begin{array}{l c l}
       \frac{y+A_{(1)}}{\iota_{1}+A_{(1)}}\cdot f_{(1)}(x) &:& y \in [y_{0},y_{1});\\
       \frac{y+A_{(2)}}{\iota_{2}+A_{(2)}}\cdot f_{(2)}(x) &:& y \in [y_{1},y_{2});\\
       \cdots & & \\
       \frac{y+A_{(K)}}{\iota_{K}+A_{(K)}}\cdot f_{(K)}(x) &:& y \in [y_{K-1},y_{K}];
     \end{array}
   \right.
   \nonumber
\end{equation}
whose limiting case, according to the Proposition \ref{pro:00}, will reach a periodic generating sequence with period $c=2$. More precisely, the Bers generating pairs will have the form
\begin{equation}
\begin{array}{l c l}
F_{0}=\sqrt{\sigma},\ G_{0}=i(\sqrt{\sigma})^{-1},\\
F_{1}=(\sqrt{\sigma})^{-1},\ G_{1}=i\sqrt{\sigma}.
\end{array}
\nonumber
\end{equation}
Performing the full numerical calculations, and considering the same values $N$, $Q$ and $P$ of the Subsection \ref{comparison}, we obtained a total error
\begin{equation}
\mathcal{E}=5.0863\times 10^{-9},
\nonumber
\end{equation}
which is slightly different that the error (\ref{analyticerror}), obtained for the limiting case where $c=1$. Nevertheless, it is possible to appreciate that the computational resources required for analysing the forward problem when $c=2$, are bigger than such required for $c=1$. Therefore, hereafter we will exclusively utilize the approach where $c=1$.
\end{remark}

\section{Examples of conductivity functions with exact representations.}
\label{exact:representation}

We will analyse a selected set of conductivity functions, for which an exact solution is known, in order to impose it as the boundary condition $u_{\textbf{c}}\vert_{\Gamma}$ to be approached. Once more, the experiments are performed within the unit circle. Notice that none of these examples posses a separable-variables form.

\subsection{The exponential case.}
\begin{proposition}
Let us consider the conductivity function 
\begin{equation}
\sigma=e^{\alpha xy},
\label{abv02}
\end{equation}
where $\alpha$ is a real constant. A particular solution of (\ref{int:00}) is
\begin{equation}
u=e^{-\alpha xy}.
\label{abv03}
\end{equation}
\end{proposition}

The Table 1 contains a condensed relation of the errors $\mathcal{E}$ when changing the number of total points per radius $P$, number of radii $Q$ and maximum number of formal powers $N$, considering $\alpha=1$. We shall notice that the behaviour of the total error does not keep a clear pattern of change, when the other parameters are modified. More precisely, the total error does not decrease monotonically when the number of points per radius $P$, or the number of radii $Q$, do it. 

On the other hand, the total error does decrease when the number $N$ of employed formal powers does. Indeed, the magnitudes of the total error indicate that this technique is appropriate for solving this boundary value problem.
\begin{table}
\centering
\begin{tabular}{| c | c | c | c |}
\hline
\tiny{\textbf{Number of formal powers}} & \tiny{\textbf{Number of radii}} & \tiny{\textbf{Number of points per radius}} & \tiny{\textbf{Total error}} \\
$N$ & $P$ & $Q$ & $\mathcal{E}$ \\
\hline
$30$ & $1000$ & $1000$ & $1.9492\times 10^{-8}$\\
\hline
$30$ & $1000$ & $800$ & $2.1979\times 10^{-8}$\\
\hline
$30$ & $1000$ & $600$ & $2.2281\times 10^{-8}$\\
\hline
$30$ & $1000$ & $400$ & $2.2156\times 10^{-8}$\\
\hline
$30$ & $1000$ & $200$ & $1.7241\times 10^{-8}$\\
\hline
$30$ & $800$ & $1000$ & $1.5221\times 10^{-8}$\\
\hline
$30$ & $600$ & $1000$ & $1.7483\times 10^{-8}$\\
\hline
$30$ & $400$ & $1000$ & $1.1651\times 10^{-8}$\\
\hline
$30$ & $200$ & $1000$ & $2.6761\times 10^{-8}$\\
\hline
$20$ & $1000$ & $1000$ & $3.2741\times 10^{-8}$\\
\hline
$10$ & $1000$ & $1000$ & $1.9312\times 10^{-7}$\\
\hline
$30$ & $500$ & $500$ & $1.1009\times 10^{-8}$\\
\hline
$20$ & $500$ & $500$ & $1.0376\times 10^{-8}$\\
\hline
$10$ & $500$ & $500$ & $1.7330\times 10^{-7}$\\
\hline
$30$ & $100$ & $100$ & $6.6772\times 10^{-7}$\\
\hline
$20$ & $100$ & $100$ & $6.9343\times 10^{-7}$\\
\hline
$10$ & $100$ & $100$ & $8.3178\times 10^{-7}$\\
\hline
$10$ & $100$ & $50$ & $7.8103\times 10^{-7}$\\
\hline
$10$ & $50$ & $50$ & $ 8.8030\times 10^{-6}$\\
\hline
$5$ & $50$ & $50$ & $  0.0317$\\
\hline
\end{tabular}\label{table:00}
\caption{Table of values corresponding to the non separable-variables exponential case $\sigma=e^{xy}$.}
\end{table}

The Table 2 contains the errors $\mathcal{E}$ considering $\alpha=5$. Their behaviour is, in general, similar to the one of the prior Table.
\begin{table}
\centering
\begin{tabular}{| c | c | c | c |}
\hline
\tiny{\textbf{Number of formal powers}} & \tiny{\textbf{Number of radii}} & \tiny{\textbf{Number of points per radius}} & \tiny{\textbf{Total error}} \\
$N$ & $P$ & $Q$ & $\mathcal{E}$ \\
\hline
$30$ & $1000$ & $1000$ & $3.3167\times 10^{-7}$\\
\hline
$30$ & $1000$ & $800$ & $3.4754\times 10^{-7}$\\
\hline
$30$ & $1000$ & $600$ & $3.0912\times 10^{-7}$\\
\hline
$30$ & $1000$ & $400$ & $3.3658\times 10^{-7}$\\
\hline
$30$ & $1000$ & $200$ & $3.3271\times 10^{-7}$\\
\hline
$30$ & $800$ & $1000$ & $2.6301\times 10^{-7}$\\
\hline
$30$ & $600$ & $1000$ & $2.2022\times 10^{-7}$\\
\hline
$30$ & $400$ & $1000$ & $5.7358\times 10^{-7}$\\
\hline
$30$ & $200$ & $1000$ & $6.4704\times 10^{-6}$\\
\hline
$20$ & $1000$ & $1000$ & $1.6141\times 10^{-6}$\\
\hline
$10$ & $1000$ & $1000$ & $0.1511$\\
\hline
$30$ & $500$ & $500$ & $3.5765\times 10^{-7}$\\
\hline
$20$ & $500$ & $500$ & $1.1817\times 10^{-6}$\\
\hline
$10$ & $500$ & $500$ & $0.1067$\\
\hline
$30$ & $100$ & $100$ & $7.2363\times 10^{-5}$\\
\hline
$20$ & $100$ & $100$ & $1.1286\times 10^{-4}$\\
\hline
$10$ & $100$ & $100$ & $0.0450$\\
\hline
$10$ & $100$ & $50$ & $0.0210$\\
\hline
$10$ & $50$ & $50$ & $0.0261$\\
\hline
$5$ & $50$ & $50$ & $ 9.4212$\\
\hline
\end{tabular}\label{table:00a}
\caption{Table of values corresponding to the non separable-variables exponential case $\sigma=e^{5xy}$.}
\end{table}

\subsection{The polynomial case.}

\begin{proposition}
Let us assume that the conductivity function has the form
\begin{equation}
\sigma=\alpha(x+y)+10,
\label{abv06}
\end{equation}
thus the function
\begin{equation}
u=\ln\left(\alpha(x+y)+10\right),
\label{abv07}
\end{equation}
will be a solution of (\ref{int:00}).
\end{proposition}

Among the other cases, this is a singular example, because utilizing a small number of points per radius $P$, and a small number of radii $Q$, together with a relatively small number of formal powers $N$, the convergence of the method is acceptable. We expose the results obtained for the case when $\alpha=1$ in the Table 3.

\begin{table}
\centering
\begin{tabular}{| c | c | c | c |}
\hline
\tiny{\textbf{Number of formal powers}} & \tiny{\textbf{Number of radii}} & \tiny{\textbf{Number of points per radius}} & \tiny{\textbf{Total error}} \\
$N$ & $P$ & $Q$ & $\mathcal{E}$ \\
\hline
$30$ & $1000$ & $1000$ & $3.6530\times 10^{-8}$\\
\hline
$30$ & $1000$ & $800$ & $3.6528\times 10^{-8}$\\
\hline
$30$ & $1000$ & $600$ & $3.6515\times 10^{-8}$\\
\hline
$30$ & $1000$ & $400$ & $3.6482\times 10^{-8}$\\
\hline
$30$ & $1000$ & $200$ & $3.6572\times 10^{-8}$\\
\hline
$30$ & $800$ & $1000$ & $4.3271\times 10^{-8}$\\
\hline
$30$ & $600$ & $1000$ & $3.5882\times 10^{-8}$\\
\hline
$30$ & $400$ & $1000$ & $2.1136\times 10^{-8}$\\
\hline
$30$ & $200$ & $1000$ & $1.1306\times 10^{-8}$\\
\hline
$20$ & $1000$ & $1000$ & $3.2499\times 10^{-8}$\\
\hline
$10$ & $1000$ & $1000$ & $5.9790\times 10^{-8}$\\
\hline
$30$ & $500$ & $500$ & $2.0991\times 10^{-8}$\\
\hline
$20$ & $500$ & $500$ & $3.1329\times 10^{-8}$\\
\hline
$10$ & $500$ & $500$ & $4.4110\times 10^{-8}$\\
\hline
$30$ & $100$ & $100$ & $2.8376\times 10^{-8}$\\
\hline
$20$ & $100$ & $100$ & $6.5198\times 10^{-8}$\\
\hline
$10$ & $100$ & $100$ & $1.1667\times 10^{-7}$\\
\hline
$10$ & $100$ & $50$ & $ 1.1973\times 10^{-7}$\\
\hline
$10$ & $50$ & $50$ & $ 1.1999\times 10^{-7}$\\
\hline
$5$ & $50$ & $50$ & $ 1.6714\times 10^{-7}$\\
\hline
$5$ & $15$ & $15$ & $ 4.0953\times 10^{-5}$\\
\hline
\end{tabular}\label{table:02}
\caption{Table of values corresponding to the non separable-variables polynomial case $\sigma=x+y+10$.}
\end{table}

When considering the case $\alpha=5$, the results presented in Table 4 preserve the properties remarked in the previous Table, even when only $5$ formal powers ($11$ base functions) are used.

\begin{table}
\centering
\begin{tabular}{| c | c | c | c |}
\hline
\tiny{\textbf{Number of formal powers}} & \tiny{\textbf{Number of radii}} & \tiny{\textbf{Number of points per radius}} & \tiny{\textbf{Total error}} \\
$N$ & $P$ & $Q$ & $\mathcal{E}$ \\
\hline
$30$ & $1000$ & $1000$ & $1.5315\times 10^{-7}$\\
\hline
$30$ & $1000$ & $800$ & $1.5330\times 10^{-7}$\\
\hline
$30$ & $1000$ & $600$ & $1.6351\times 10^{-7}$\\
\hline
$30$ & $1000$ & $400$ & $1.5171\times 10^{-7}$\\
\hline
$30$ & $1000$ & $200$ & $1.5303\times 10^{-7}$\\
\hline
$30$ & $800$ & $1000$ & $1.0403\times 10^{-7}$\\
\hline
$30$ & $600$ & $1000$ & $9.3654\times 10^{-8}$\\
\hline
$30$ & $400$ & $1000$ & $5.1980\times 10^{-8}$\\
\hline
$30$ & $200$ & $1000$ & $3.3674\times 10^{-8}$\\
\hline
$20$ & $1000$ & $1000$ & $2.3294\times 10^{-7}$\\
\hline
$10$ & $1000$ & $1000$ & $2.8096\times 10^{-4}$\\
\hline
$30$ & $500$ & $500$ & $6.6396\times 10^{-8}$\\
\hline
$20$ & $500$ & $500$ & $9.4135\times 10^{-8}$\\
\hline
$10$ & $500$ & $500$ & $1.9824\times 10^{-4}$\\
\hline
$30$ & $100$ & $100$ & $2.6497\times 10^{-7}$\\
\hline
$20$ & $100$ & $100$ & $4.1895\times 10^{-7}$\\
\hline
$10$ & $100$ & $100$ & $8.3724\times 10^{-5}$\\
\hline
$10$ & $100$ & $50$ & $  3.1943\times 10^{-5}$\\
\hline
$10$ & $50$ & $50$ & $ 4.7861\times 10^{-5}$\\
\hline
$5$ & $50$ & $50$ & $ 0.0145$\\
\hline
$5$ & $15$ & $15$ & $ 0.0039$\\
\hline
\end{tabular}\label{table:02a}
\caption{Table of values corresponding to the non separable-variables polynomial case $\sigma=5(x+y)+10$.}
\end{table}

\subsection{The Lorentzian case.}
\label{lore:case}
\begin{proposition}
Let the conductivity function have the form
\begin{equation}
\sigma=\frac{1}{\left(x+y\right)^{2}+\alpha}.
\label{abv04}
\end{equation}
An exact solution for the equation (\ref{int:00}) is
\begin{equation}
u=\frac{\left(x+y\right)^{3}}{3}+\alpha\left(x+y\right).
\label{abv05}
\end{equation}
\end{proposition}

We will analyse the case when the constant $\alpha=1$, showing the obtained results in the Table 5. Once more, we detect that decrementing the number of points per radius $P$, as well the number of radii $Q$, does not necessarily provoke a diminution of the total error $\mathcal{E}$.

\begin{table}
\centering
\begin{tabular}{| c | c | c | c |}
\hline
\tiny{\textbf{Number of formal powers}} & \tiny{\textbf{Number of radii}} & \tiny{\textbf{Number of points per radius}} & \tiny{\textbf{Total error}} \\
$N$ & $P$ & $Q$ & $\mathcal{E}$ \\
\hline
$30$ & $1000$ & $1000$ & $1.1213\times 10^{-8}$\\
\hline
$30$ & $1000$ & $800$ & $1.9372\times 10^{-8}$\\
\hline
$30$ & $1000$ & $600$ & $1.7023\times 10^{-8}$\\
\hline
$30$ & $1000$ & $400$ & $1.9995\times 10^{-8}$\\
\hline
$30$ & $1000$ & $200$ & $1.9013\times 10^{-8}$\\
\hline
$30$ & $800$ & $1000$ & $1.6567\times 10^{-8}$\\
\hline
$30$ & $600$ & $1000$ & $1.8385\times 10^{-8}$\\
\hline
$30$ & $400$ & $1000$ & $1.2829\times 10^{-8}$\\
\hline
$30$ & $200$ & $1000$ & $3.7494\times 10^{-8}$\\
\hline
$20$ & $1000$ & $1000$ & $1.1449\times 10^{-7}$\\
\hline
$10$ & $1000$ & $1000$ & $4.7292\times 10^{-4}$\\
\hline
$30$ & $500$ & $500$ & $1.3655\times 10^{-8}$\\
\hline
$20$ & $500$ & $500$ & $7.4664\times 10^{-8}$\\
\hline
$10$ & $500$ & $500$ & $5.5399\times 10^{-4}$\\
\hline
$30$ & $100$ & $100$ & $5.9646\times 10^{-7}$\\
\hline
$20$ & $100$ & $100$ & $6.8813\times 10^{-7}$\\
\hline
$10$ & $100$ & $100$ & $4.5270\times 10^{-4}$\\
\hline
$10$ & $100$ & $50$ & $ 4.5305\times 10^{-4}$\\
\hline
$10$ & $50$ & $50$ & $  4.1998\times 10^{-4}$\\
\hline
$5$ & $50$ & $50$ & $  0.0125$\\
\hline
\end{tabular}\label{table:03}
\caption{Table of values corresponding to the non separable-variables Lorentzial case $\sigma=\left(\left(x+y\right)^{2}+1 \right)^{-1}$.}
\end{table}

The case when $\alpha=0.01$, in (\ref{abv04}) and (\ref{abv05}), does provide additional information for our study. The maximum value of the conductivity $\sigma$, at the origin, is ten thousands times bigger than the one belonging to the cases studied before. From this point of view, the notorious increment of the total errors is justified. Yet, it is interesting that the diminution of $P$ and $Q$ do not influence that much in the errors. The Table 6 contains the numerical data that sustain our affirmations.

\begin{table}
\centering
\begin{tabular}{| c | c | c | c |}
\hline
\tiny{\textbf{Number of formal powers}} & \tiny{\textbf{Number of radii}} & \tiny{\textbf{Number of points per radius}} & \tiny{\textbf{Total error}} \\
$N$ & $P$ & $Q$ & $\mathcal{E}$ \\
\hline
$30$ & $1000$ & $1000$ & $0.1671$\\
\hline
$30$ & $1000$ & $800$ & $0.1671$\\
\hline
$30$ & $1000$ & $600$ & $0.1671$\\
\hline
$30$ & $1000$ & $400$ & $0.1671$\\
\hline
$30$ & $1000$ & $200$ & $0.1672$\\
\hline
$30$ & $800$ & $1000$ & $0.1491$\\
\hline
$30$ & $600$ & $1000$ & $0.1283$\\
\hline
$30$ & $400$ & $1000$ & $0.1030$\\
\hline
$30$ & $200$ & $1000$ & $0.0662$\\
\hline
$20$ & $1000$ & $1000$ & $0.6420$\\
\hline
$10$ & $1000$ & $1000$ & $3.1839$\\
\hline
$30$ & $500$ & $500$ & $0.1165$\\
\hline
$20$ & $500$ & $500$ & $0.4509$\\
\hline
$10$ & $500$ & $500$ & $2.2466$\\
\hline
$30$ & $100$ & $100$ & $0.0265$\\
\hline
$20$ & $100$ & $100$ & $0.1559$\\
\hline
$10$ & $100$ & $100$ & $0.9435$\\
\hline
$10$ & $100$ & $50$ & $0.0870$\\
\hline
$10$ & $50$ & $50$ & $0.5879$\\
\hline
$5$ & $50$ & $50$ & $120.3691$\\
\hline
\end{tabular}\label{table:03a}
\caption{Table of values corresponding to the non separable-variables Lorentzian case $\sigma=\left(\left(x+y\right)^{2}+0.01\right)^{-1}$.}
\end{table}

\subsection{The sinusoidal case.}

This example has been selected because it proved to be one of the most challenging cases for testing the performance of the numerical method. Our analysis will be based in the following proposition.

\begin{proposition}
Let us consider the sinusoidal conductivity
\begin{equation}
\sigma=1+\sin \alpha xy.
\label{abv08}
\end{equation}
We can verify by direct substitution that the function
\begin{equation}
u=\left(\tan\left(\frac{\alpha xy}{2}\right)+1\right)^{-1},
\label{abv09}
\end{equation}
is a particular solution of (\ref{int:00}).
\end{proposition}

Let us focus our attention into the case $\alpha=1$. The Table 7 contains what it could be considered a set of acceptable approaches for the boundary value problem, displaying a similar behaviour to the rest of cases previously shown. More interesting information reaches when considering $\alpha=5$ in the argument of the sinusoidal expression (\ref{abv08}), but keeping the value $\alpha=1$ in the boundary condition (\ref{abv09}). This is mainly because the condition is not defined at several points of the boundary $\Gamma$, when considering $\alpha=5$.

Even not any formal extension of the Theorem \ref{TheoremComplete} is known for the limiting case posed in the Proposition \ref{pro:00}, about the completeness of the set (\ref{eie:03}), one could expect a diminution of the total error $\mathcal{E}$ when increasing the number $N$ of formal powers, as registered in the other examples. But the Table 8 shows that, for this particular case, the behaviour is the opposite. This implies that we have found an example that could be very significant to adequately understand the properties of the formal powers constructed by virtue of the Proposition \ref{pro:00}. The case must certainly be analysed with more detail in further works.
\begin{table}
\centering
\begin{tabular}{| c | c | c | c |}
\hline
\tiny{\textbf{Number of formal powers}} & \tiny{\textbf{Number of radii}} & \tiny{\textbf{Number of points per radius}} & \tiny{\textbf{Total error}} \\
$N$ & $P$ & $Q$ & $\mathcal{E}$ \\
\hline
$30$ & $1000$ & $1000$ & $1.2451\times 10^{-8}$\\
\hline
$30$ & $1000$ & $800$ & $1.2263\times 10^{-8}$\\
\hline
$30$ & $1000$ & $600$ & $1.1838\times 10^{-8}$\\
\hline
$30$ & $1000$ & $400$ & $1.3812\times 10^{-8}$\\
\hline
$30$ & $1000$ & $200$ & $1.4932\times 10^{-8}$\\
\hline
$30$ & $800$ & $1000$ & $9.4572\times 10^{-9}$\\
\hline
$30$ & $600$ & $1000$ & $1.1218\times 10^{-8}$\\
\hline
$30$ & $400$ & $1000$ & $9.7996\times 10^{-9}$\\
\hline
$30$ & $200$ & $1000$ & $9.6253\times 10^{-9}$\\
\hline
$20$ & $1000$ & $1000$ & $2.3682\times 10^{-8}$\\
\hline
$10$ & $1000$ & $1000$ & $3.9153\times 10^{-5}$\\
\hline
$30$ & $500$ & $500$ & $9.7774\times 10^{-9}$\\
\hline
$20$ & $500$ & $500$ & $2.4493\times 10^{-8}$\\
\hline
$10$ & $500$ & $500$ & $5.6988\times 10^{-5}$\\
\hline
$30$ & $100$ & $100$ & $3.7062\times 10^{-7}$\\
\hline
$20$ & $100$ & $100$ & $5.1172\times 10^{-7}$\\
\hline
$10$ & $100$ & $100$ & $9.3918\times 10^{-5}$\\
\hline
$10$ & $100$ & $50$ & $ 9.3414\times 10^{-5}$\\
\hline
$10$ & $50$ & $50$ & $  6.9851\times 10^{-5}$\\
\hline
$5$ & $50$ & $50$ & $  0.0497$\\
\hline
\end{tabular}\label{table:01}
\caption{Table of values corresponding to the non separable-variables sinusoidal case $\sigma=1+\sin xy$.}
\end{table}

Finally, we shall remark that 
\begin{equation}
\sigma=1+\sin 5xy>0,\ \forall x,y\in\Omega.
\end{equation}
\begin{table}
\centering
\begin{tabular}{| c | c | c | c |}
\hline
\tiny{\textbf{Number of formal powers}} & \tiny{\textbf{Number of radii}} & \tiny{\textbf{Number of points per radius}} & \tiny{\textbf{Total error}} \\
$N$ & $P$ & $Q$ & $\mathcal{E}$ \\
\hline
$30$ & $1000$ & $1000$ & $1.0338\times 10^{4}$\\
\hline
$30$ & $1000$ & $800$ & $1.2828\times 10^{5}$\\
\hline
$30$ & $1000$ & $600$ & $7.1245\times 10^{3}$\\
\hline
$30$ & $1000$ & $400$ & $2.7315\times 10^{4}$\\
\hline
$30$ & $1000$ & $200$ & $1.4515\times 10^{4}$\\
\hline
$30$ & $800$ & $1000$ & $2.7273\times 10^{4}$\\
\hline
$30$ & $600$ & $1000$ & $7.1022\times 10^{3}$\\
\hline
$30$ & $400$ & $1000$ & $2.2060\times 10^{3}$\\
\hline
$30$ & $200$ & $1000$ & $1.9391\times 10^{3}$\\
\hline
$20$ & $1000$ & $1000$ & $3.6205\times 10^{4}$\\
\hline
$10$ & $1000$ & $1000$ & $2.7627\times 10^{5}$\\
\hline
$30$ & $500$ & $500$ & $2.3643\times 10^{3}$\\
\hline
$20$ & $500$ & $500$ & $1.2936\times 10^{4}$\\
\hline
$10$ & $500$ & $500$ & $7.0407\times 10^{4}$\\
\hline
$30$ & $100$ & $100$ & $3.2401$\\
\hline
$20$ & $100$ & $100$ & $116.1873$\\
\hline
$10$ & $100$ & $100$ & $2.2541\times 10^{3}$\\
\hline
$10$ & $100$ & $50$ & $  0.1006$\\
\hline
$10$ & $50$ & $50$ & $  24.2426$\\
\hline
$5$ & $50$ & $50$ & $  47.5417$\\
\hline
\end{tabular}\label{table:01a}
\caption{Table of values corresponding to the non separable-variables sinusoidal case $\sigma=1+\sin 5xy$.}
\end{table}

\section{Conductivities corresponding to geometrical distributions.}

This section is dedicated to analyse what it could well be considered one of the most important contributions of this work to the State of the Art for solving the Dirichlet boundary value problem of (\ref{int:00}), employing the elements of the Pseudoanalytic Function Theory: The study of conductivity distributions rising from geometrical cases.

As it was mentioned before, the tools provided in \cite{bers} and \cite{kpa} impose an important condition, seldom fulfilled in physical experimental models: The conductivity $\sigma$ must be a separable-variables function. The Proposition \ref{pro:00} allows us to overpass this restriction by stating that any conductivity function:
\begin{equation}
\sigma:\Omega\left(\mathbb{R}^{2}\right)\rightarrow\mathbb{R},
\label{sigma:map}
\end{equation}
can be considered the limiting case of a piecewise-defined function, of the form (\ref{abv01}). The Section above studied the cases when the $\sigma$ possesses an exact representation, but the mapping (\ref{sigma:map}) indicates that it is enough to know the value of $\sigma$ at every point $(x,y)$ within the domain $\Omega$. Thus, the conductivities given by geometrical distributions, can be treated with the same tools that we employed in the Section \ref{exact:representation}.

Before continuing with this branch of the analysis, we shall remark that, as it was explained in Subsection \ref{highacc}, when employing the expression (\ref{high}) for numerically approaching the formal powers, we are indeed performing a piecewise interpolating polynomial method of degree $1$, on every pair of valued-points on the radius where the conductivity $\sigma$ is defined. 

This implies that the discontinuities of the first kind, appearing in the conductivity over one radius, will not be taken into account, since there will always exist a line with finite slope, relating the values $\sigma(r[p])$ and $\sigma(r[p+1])$, between which the discontinuity is found.

Nevertheless, the set of examples presented in the upcoming Subsections, will show that this approach provides numerical solutions with an acceptable degree of accuracy.

\subsection{The concentric disks.}

Let us consider a piecewise conductivity function, in polar coordinates, of the form:
\begin{equation}
   \sigma (r) = \left\{
     \begin{array}{l c l}
       100 &:&  r \in [0,0.2);\\
       30 &:&  r \in [0.2,0.4);\\
       20 &:&  r \in [0.4,0.6);\\
       15 &:&  r \in [0.6,0.8);\\
       10 &:&  r \in [0.8,1].
     \end{array}
   \right.
   \label{ns:08}   
\end{equation}

Here $r$ denotes the radius. The Figure 2 illustrates this case. To select a boundary condition for the conductivities defined by geometrical figures, without performing physical measurements, it is not a trivial task. Nevertheless, the conductivity defined in (\ref{ns:08}) may be related with the separable-variables Lorentzian case, studied in Proposition \ref{prop02}. For this reason, the boundary condition will be a variation of the expression (\ref{bs10}):
\begin{equation}
u=\frac{x^3+y^3}{3}+0.01\left(x+y\right).
\nonumber
\end{equation}
\begin{figure}
\centering
\includegraphics[scale=0.30]{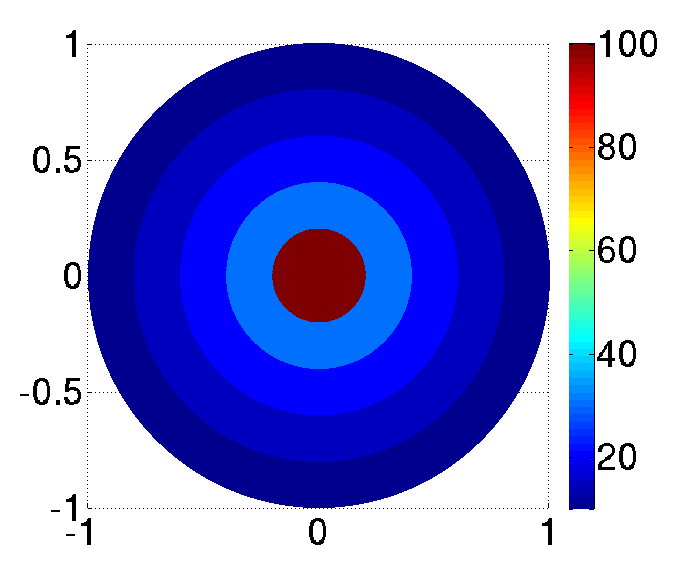}
\label{fig:ns04}
\caption{Conductivity defined according to the expression (\ref{ns:08}).}
\end{figure}

The Table 9 summarizes the total errors $\mathcal{E}$ obtained when changing the number of points $P$, the total radii $Q$, and the employed formal powers $N$. The behaviour of the total error is very similar to the one exhibited in the Lorentzian case of the Subsection \ref{lore:case}.

\begin{table}
\centering
\begin{tabular}{| c | c | c | c |}
\hline
\tiny{\textbf{Number of formal powers}} & \tiny{\textbf{Number of radii}} & \tiny{\textbf{Number of points per radius}} & \tiny{\textbf{Total error}} \\
$N$ & $P$ & $Q$ & $\mathcal{E}$ \\
\hline
$40$ & $1000$ & $1000$ & $3.6234\times 10^{-9}$\\
\hline
$40$ & $1000$ & $800$ & $2.7852\times 10^{-9}$\\
\hline
$40$ & $1000$ & $600$ & $2.5213\times 10^{-9}$\\
\hline
$40$ & $1000$ & $400$ & $2.3199\times 10^{-9}$\\
\hline
$40$ & $1000$ & $200$ & $1.4331\times 10^{-8}$\\
\hline
$40$ & $800$ & $1000$ & $3.6234\times 10^{-9}$\\
\hline
$40$ & $600$ & $1000$ & $3.6234\times 10^{-9}$\\
\hline
$40$ & $400$ & $1000$ & $3.6234\times 10^{-9}$\\
\hline
$40$ & $200$ & $1000$ & $3.6234\times 10^{-9}$\\
\hline
$20$ & $1000$ & $1000$ & $3.3615\times 10^{-9}$\\
\hline
$20$ & $500$ & $500$ & $2.2764\times 10^{-9}$\\
\hline
$40$ & $100$ & $100$ & $2.4841\times 10^{-7}$\\
\hline
$20$ & $100$ & $100$ & $3.3996\times 10^{-7}$\\
\hline
$10$ & $50$ & $50$ & $  8.2633\times 10^{-6}$\\
\hline
$5$ & $50$ & $50$ & $  1.1721\times 10^{-5}$\\
\hline
\end{tabular}\label{table:04}
\caption{Table of values corresponding to the conductivity $\sigma$ of the Figure 2: The concentric disks.}
\end{table}

\subsection{Disk out of the center, within the unit circle.}

Let us consider now the conductivity function illustrated in the Figure 3. The blue section represents conductivity values of $\sigma=10$, whereas the red area expresses $\sigma=100$. The red disk corresponds to the equation
\begin{equation}
(x-0.6)^2+y^2\leq 0.2.
\nonumber
\end{equation}

For this case, and considering once more the possible similarity with the Lorentzian case (\ref{bs10}), we will impose the boundary condition:
\begin{equation}
u=\frac{\left(x-0.6\right)^{3}+y^{3}}{3}+0.01\left(x-0.6+y\right).
\nonumber
\end{equation}

The results of the numerical calculations are displayed in the Table 10. It is remarkable that the total errors are, in general, five thousands times bigger than such shown in the Table 9, where the concentric disks were analysed.

\begin{figure}
\centering
\includegraphics[scale=0.25]{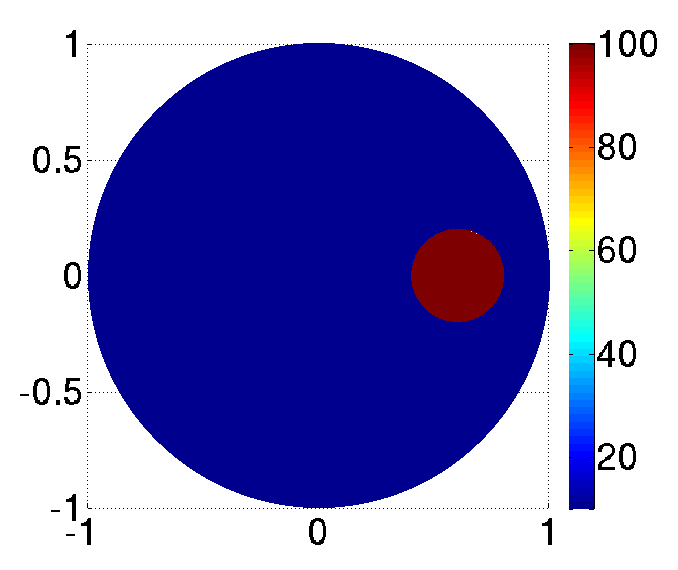}
\label{fig:ns06b}
\caption{Conductivity function with one disk within the unit circle. The red surface represent $\sigma=100$, whereas the blue sections denote $\sigma=10$.}
\end{figure}

\begin{table}
\centering
\begin{tabular}{| c | c | c | c |}
\hline
\tiny{\textbf{Number of formal powers}} & \tiny{\textbf{Number of radii}} & \tiny{\textbf{Number of points per radius}} & \tiny{\textbf{Total error}} \\
$N$ & $P$ & $Q$ & $\mathcal{E}$ \\
\hline
$40$ & $1000$ & $1000$ & $7.8082\times 10^{-4}$\\
\hline
$40$ & $1000$ & $800$ & $6.9845\times 10^{-4}$\\
\hline
$40$ & $1000$ & $600$ & $5.9861\times 10^{-4}$\\
\hline
$40$ & $1000$ & $400$ & $4.3758\times 10^{-4}$\\
\hline
$40$ & $1000$ & $200$ & $3.1089\times 10^{-4}$\\
\hline
$40$ & $800$ & $1000$ & $7.8048\times 10^{-4}$\\
\hline
$40$ & $600$ & $1000$ & $7.7999\times 10^{-4}$\\
\hline
$40$ & $400$ & $1000$ & $7.6313\times 10^{-4}$\\
\hline
$40$ & $200$ & $1000$ & $7.8686\times 10^{-4}$\\
\hline
$20$ & $1000$ & $1000$ & $1.6829\times 10^{-3}$\\
\hline
$20$ & $500$ & $500$ & $1.8330\times 10^{-3}$\\
\hline
$40$ & $100$ & $100$ & $6.1065\times 10^{-5}$\\
\hline
$20$ & $100$ & $100$ & $7.7260\times 10^{-4}$\\
\hline
$10$ & $50$ & $50$ & $  1.5540\times 10^{-3}$\\
\hline
$5$ & $50$ & $50$ & $  3.9329\times 10^{-3}$\\
\hline
\end{tabular}
\label{table:05}
\caption{Table of values corresponding to the conductivity $\sigma$ of the Figure 3: One disk out of the center of the unit circle.}
\end{table}

\subsection{Square within the unit circle.}
\label{square:case}

\begin{figure}
\centering
\includegraphics[scale=0.18]{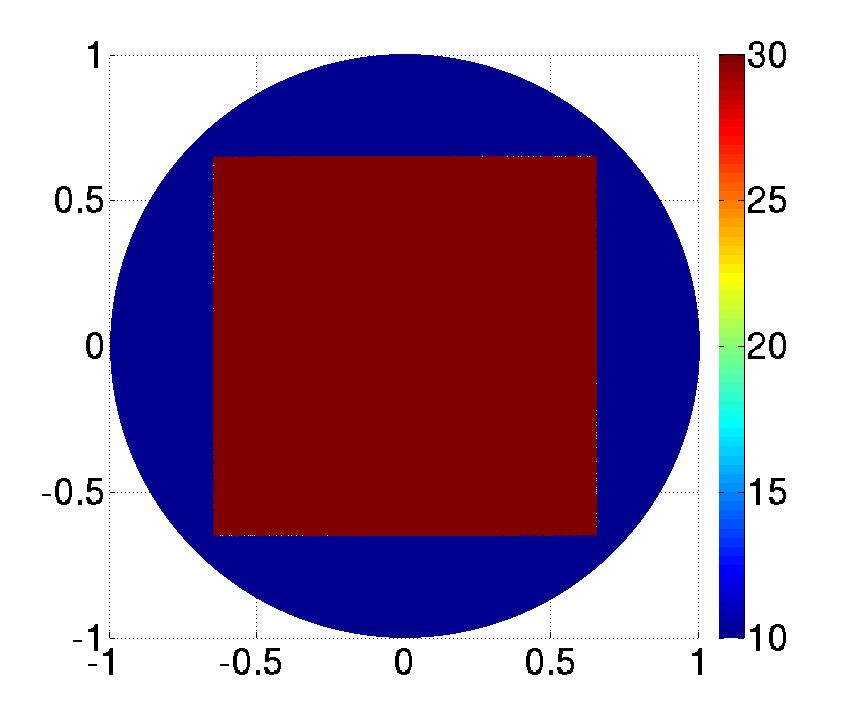}
\caption{Square within the unit circle. The blue surface represents a conductivity $\sigma=10$, and the red indicates $\sigma=100$.}
\end{figure}

As plotted in Figure 4, inside the unitary circle we located a perfect square, whose superior and inferior sides are parallel to the $x$-axis, whereas its left and right sides are parallel to the $y$-axis, being their lengths all equal to $0.65$. The square corners are equidistant to the origin, and we have fixed one point of the set (\ref{bs:10aa}) at each corner, so they can effectively take part into the calculations of the formal powers. For this case, we will impose a boundary condition of the form:

\begin{equation}
u=\frac{x^{3}+y^{3}}{3}+0.1\left(x+y\right).
\label{square:circle}
\end{equation}

The obtained errors are shown in Table 11, and even they are, in general, one thousand times bigger than the errors calculated for the disk out of the center, presented in the previous Subsection, it is important to remark that the presence of the corners did not influence significantly the convergence of the approach on the boundary $\Gamma$.

This implies that, for this conductivity distribution, with the boundary condition (\ref{square:circle}), not any supplementary regularization method is required.

As a matter of fact, the complete following Section is composed by a collection of cases where not any additional regularization methods are needed to successfully approach the boundary condition, in non-smooth bounded domains.

\begin{table}
\centering
\begin{tabular}{| c | c | c | c |}
\hline
\tiny{\textbf{Number of formal powers}} & \tiny{\textbf{Number of radii}} & \tiny{\textbf{Number of points per radius}} & \tiny{\textbf{Total error}} \\
$N$ & $P$ & $Q$ & $\mathcal{E}$ \\
\hline
$40$ & $1000$ & $1000$ & $1.4598\times 10^{-2}$\\
\hline
$40$ & $1000$ & $800$ & $1.4603\times 10^{-2}$\\
\hline
$40$ & $1000$ & $600$ & $1.4601\times 10^{-2}$\\
\hline
$40$ & $1000$ & $400$ & $1.4619\times 10^{-2}$\\
\hline
$40$ & $1000$ & $200$ & $1.4638\times 10^{-2}$\\
\hline
$40$ & $800$ & $1000$ & $1.4160\times 10^{-2}$\\
\hline
$40$ & $600$ & $1000$ & $4.0513\times 10^{-2}$\\
\hline
$40$ & $400$ & $1000$ & $2.7287\times 10^{-2}$\\
\hline
$40$ & $200$ & $1000$ & $1.4638\times 10^{-2}$\\
\hline
$20$ & $1000$ & $1000$ & $2.3516\times 10^{-2}$\\
\hline
$20$ & $500$ & $500$ & $1.6912\times 10^{-2}$\\
\hline
$40$ & $100$ & $100$ & $3.2869\times 10^{-3}$\\
\hline
$20$ & $100$ & $100$ & $3.4552\times 10^{-2}$\\
\hline
$10$ & $50$ & $50$ & $  6.8926\times 10^{-2}$\\
\hline
$5$ & $50$ & $50$ & $  1.9362\times 10^{-1}$\\
\hline
\end{tabular}
\label{table:06}
\caption{Table of values corresponding to the conductivity $\sigma$ of the Figure 4: A square within the unit circle.}
\end{table}

\section{Brief analysis of a non-circular domains.}
\label{pollito}

Many numerical methods, for solving partial differential equations at some domain $\Omega$ in the plane, often require regularization techniques when the derivative of the parametric curve describing $\Gamma$, possesses discontinuities (see \emph{e.g.} \cite{kond}). In general, the method posed in this work is not the exception. Nevertheless, there exist some particular cases such that not any regularization procedures are required, for adequately approaching the boundary condition in non-smooth domains. 

The examples we present in the forward paragraphs were heuristically found, and it is important to clarify that the convergence of the approached solutions was only achieved when establishing a certain set of values. More precisely, we fixed the number of radii $Q=100$, whereas the number of points per radius $P=100$.

It is interesting to remark that, if increasing the number $Q$ in the selected examples, the convergence at the non-smooth regions will not be achieved. Neither will do if decrementing it. Moreover, we found out that the number of formal powers $2N+1$ employed in the approach must be equal to $91$ ($46$ formal powers with coefficient $a_{n}=1$ and $45$ with coefficient $a_{n}=i$), for obtaining the minimum total error $\mathcal{E}$. This will be valid for all examples shown below.

This Section is included mainly for remarking that, the numerical solving of Dirichlet boundary value problems for elliptic equations on the plane, employing the elements of the Pseudoanalytic Function Theory, possesses the property of approaching adequately the boundary condition in non-smooth domains, when certain conditions are fulfilled. This motivate a deeper study of such property, for correctly understanding which are the conditions to warrant the convergence, and for which cases this can be effectively done.

We shall analyse the behaviour of the method on the boundary $\Gamma$ of a non-circular domain, that can be described as follows:
\begin{enumerate}
\item The domain $\Omega$ coincides with the unit circle at the left-hand side of the $y$-axis parallel line $x=\cos\frac{\pi}{10}$.
\item At the right-hand side of this axis, the domain $\Omega$ is the area bounded by the line segments:
\begin{equation}
\begin{array}{c}
y=-0.5629x+ 0.8443,\\
y=0.5629x+ 0.8443;
\end{array}
\label{lines}
\end{equation}
in the closed interval $x\in\left[\cos \frac{\pi}{10},1.5\right].$
\end{enumerate}

\subsection{The Lorentzian case.}

Let us consider a conductivity $\sigma$ within the domain described in the Subsection \ref{pollito}, with the form:
\begin{equation}
\sigma=\left(\frac{1}{x^{2}+0.1}\right)\left(\frac{1}{y^{2}+0.1}\right),
\label{newsigma}
\end{equation}
imposing as the boundary condition the exact solution:
\begin{equation}
u=\frac{\left(x^{3}+y^{3}\right)}{3}+0.1\left(x+y\right).
\label{new:lor}
\end{equation}
The illustration of this example is displayed in Figure 5.

\begin{figure}
\centering
\includegraphics[scale=0.3]{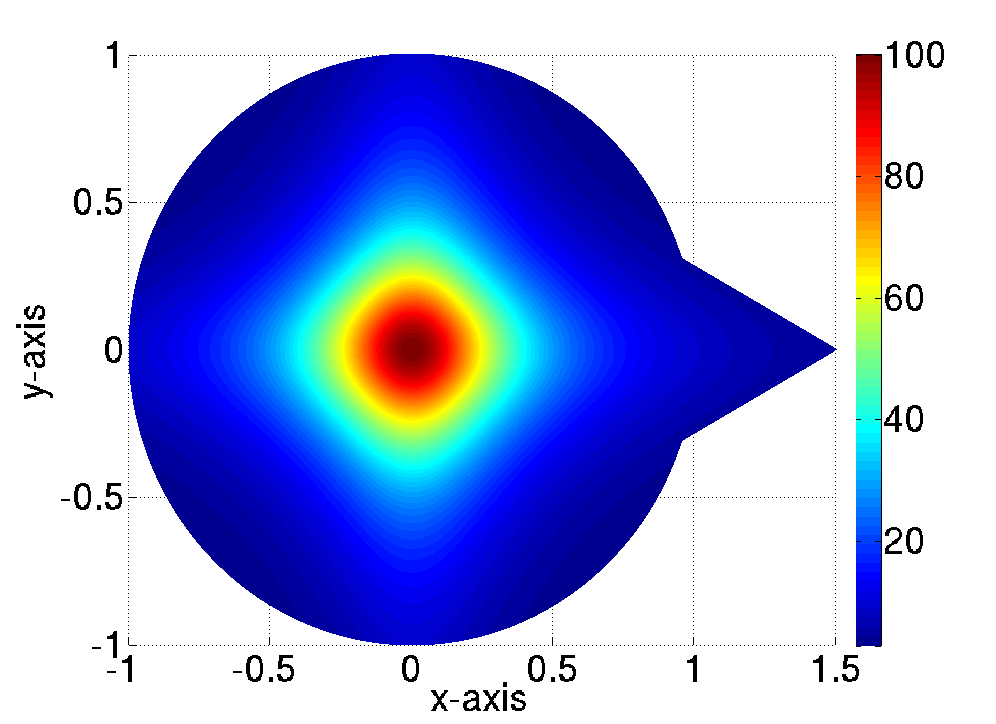}
\caption{Example of a conductivity $\sigma=\left(x^{2}+0.1\right)^{-1}\left(y^{2}+0.1\right)^{-1}$, whithin a non-circular domain.}
\end{figure}

As mentioned above, we shall fix the number of radii $Q=100$, and the number of points per radius $P=100$.
The Figure 6 displays the approach of the boundary condition (\ref{new:lor}) using $2N+1=91$ base functions. We remark that the Figure 6 was traced keeping an angular perspective, broaching an angle $\theta\in\left[-\pi,\pi\right]$, represented by the ordinate-axis, whereas the abscissa-axis plots the electric potential values related to each angle. The total error $\mathcal{E}$ for this example was:
\begin{equation}
\mathcal{E}=1.8355\times 10^{-4}.
\nonumber
\end{equation}

We also mention that one of the points where the boundary condition (\ref{new:lor}) has been valued, was located precisely at the intersection of the line segments defined in (\ref{lines}): $(1.5,0)$, where we find the most representative non-smoothness of the boundary $\Gamma$. This point coincides with the radius corresponding to the angle $\omega=0$. 

Another two boundary value points were fixed at $(1.5-\cos\frac{\pi}{10},\sin\frac{\pi}{10})$ and $(1.5-\cos\frac{\pi}{10},-\sin\frac{\pi}{10})$, where the derivative of the parametric curve describing the boundary $\Gamma$, is also discontinuous. These points are associated with the radii located at the angles $\omega=\frac{\pi}{10}$ and $\omega=-\frac{\pi}{10}$.

\begin{figure}
\centering
\includegraphics[scale=0.30]{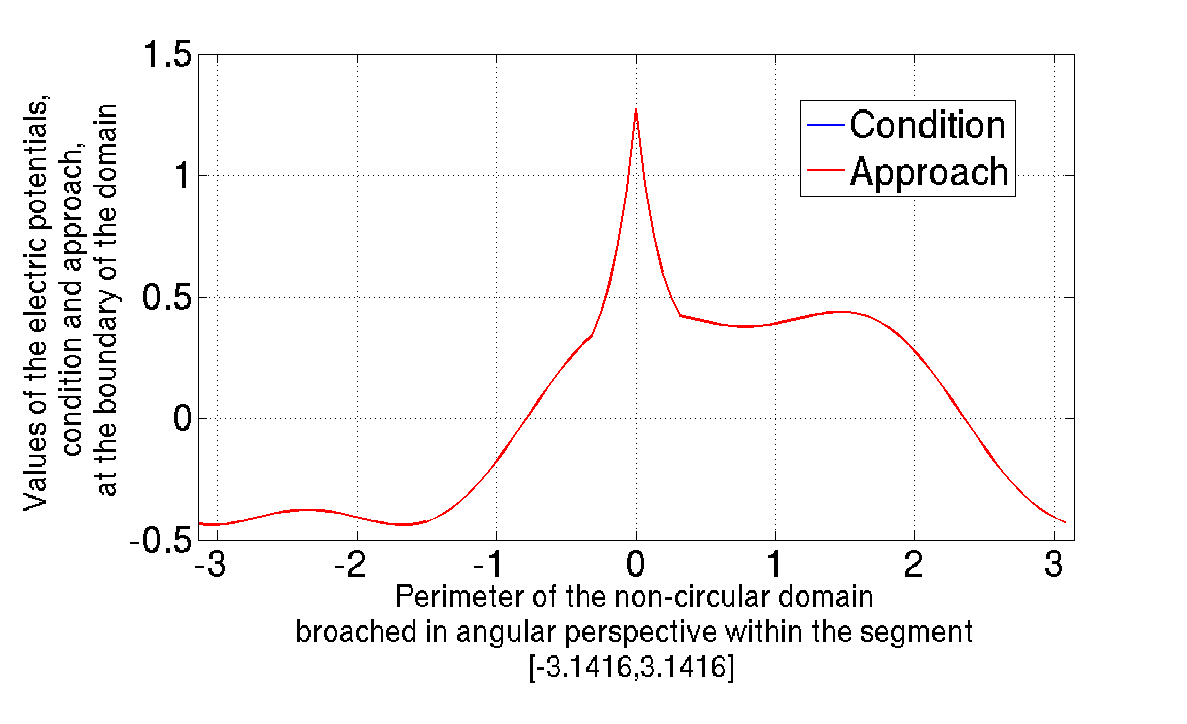}
\caption{Comparison between the boundary condition (\ref{new:lor}) and the numerical approach employing $2N+1=91$ base functions, for the case when $\sigma$ possesses the form $\sigma=\left(x^{2}+0.1\right)^{-1}\left(y^{2}+0.1\right)^{-1}$.}
\end{figure}

When reducing the number of base functions at $2N+1=51$, we observe a considerable divergence of the approached solution, around the boundary value related to the angle $\omega=0$, as it is displayed in Figure 7. Once more, we remark that one boundary value point was fixed at $(1.5,0)$. On the other hand, since the distribution of the boundary value points is performed according to the angular expression (\ref{angle:set}), the points $(1.5-\cos\frac{\pi}{10},\sin\frac{\pi}{10})$ and $(1.5-\cos\frac{\pi}{10},-\sin\frac{\pi}{10})$ were not included in the analysis of this example.

The total error employing $2N+1=51$ base functions was:
\begin{equation}
\mathcal{E}=0.3869.
\nonumber
\end{equation}

\begin{figure}
\centering
\includegraphics[scale=0.30]{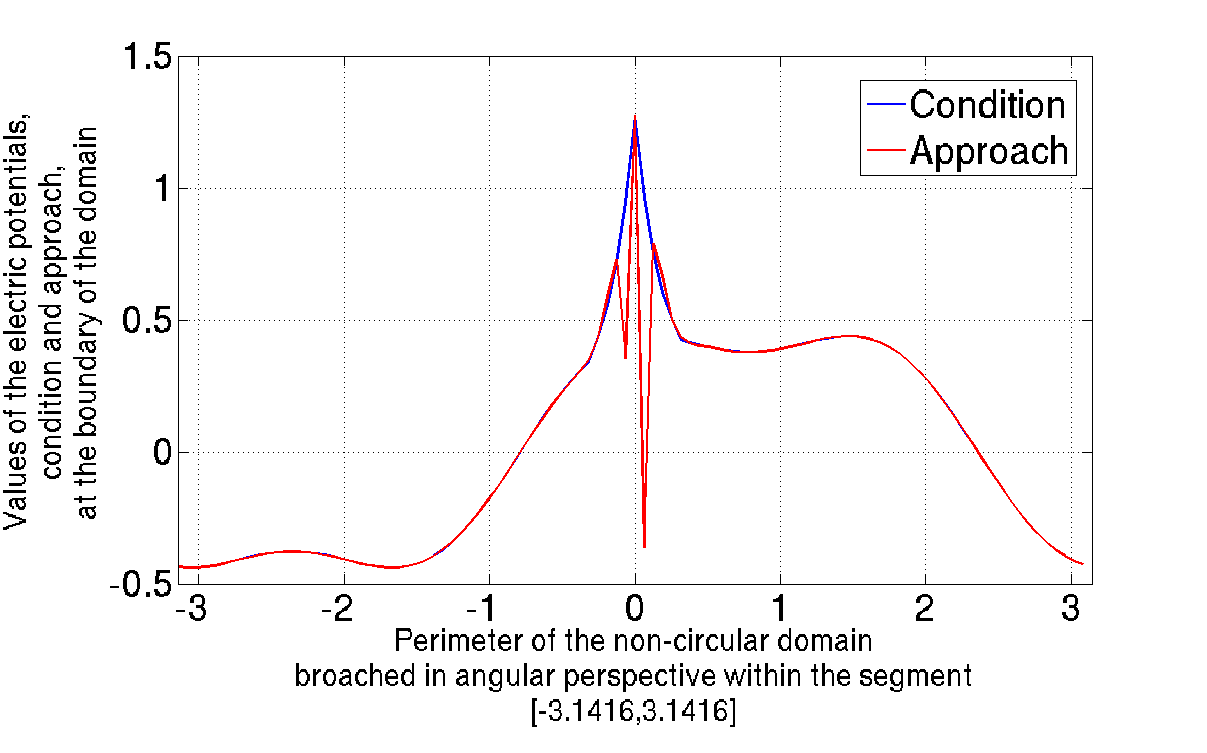}
\caption{Comparison between the boundary condition (\ref{new:lor}) and the numerical approach employing $2N+1=51$ base functions, for the case when $\sigma$ possesses the form $\sigma=\left(x^{2}+0.1\right)^{-1}\left(y^{2}+0.1\right)^{-1}$.}
\end{figure}

\subsection{Concentric circles.}

We shall consider a geometrical variation of the conductivity distribution introduced in (\ref{ns:08}), within the same domain $\Omega$ described at the beginning of the Section \ref{pollito}, as shown in Figure 8. For this example, we will perform the initial experiment taking into account $2N+1=91$ base functions. We also fix three boundary value points at the discontinuities of the derivative of $\Gamma$ indicated before.

\begin{figure}
\centering
\includegraphics[scale=0.3]{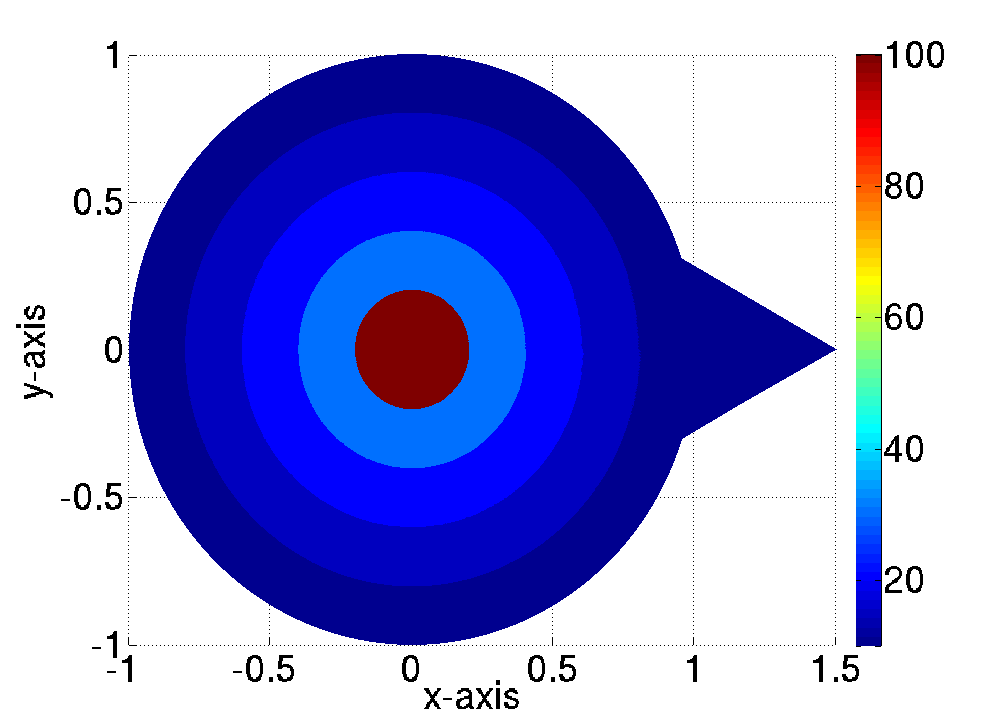}
\caption{Example of a conductivity $\sigma$ composed by concentric circles in a non-circular domain.}
\end{figure}

We choose to impose the boundary condition employed in (\ref{new:lor}), and the Figure 9 illustrates the behaviour of the approached solution when compared with the boundary condition. The total error was:
\begin{equation}
\mathcal{E}=3.4217\times 10^{-4}.
\nonumber
\end{equation}

\begin{figure}
\centering
\includegraphics[scale=0.30]{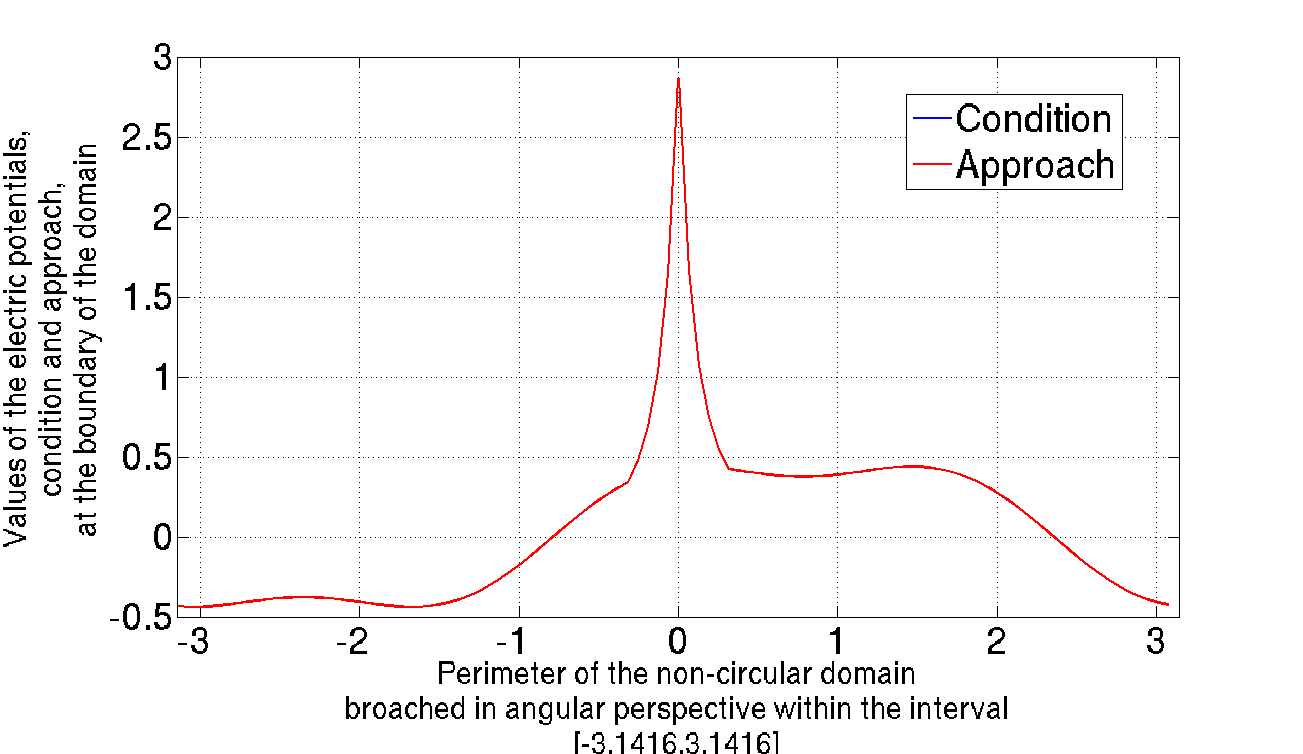}
\caption{Comparison between the boundary condition (\ref{new:lor}) and the numerical approach employing $2N+1=91$ base functions, for the case when $\sigma$ possesses the form plotted in Figure 8.}
\end{figure}

The second experiment took into account only $2N+1=51$ base functions. Again, the boundary value points were distributed at the perimeter of the domain $\Omega$ according to the angular model (\ref{angle:set}), empathizing that one point was fixed at $(1.5,0)$. As displayed in figure 10, the highest divergence of the approached solution is located around this point.

\begin{figure}
\centering
\includegraphics[scale=0.30]{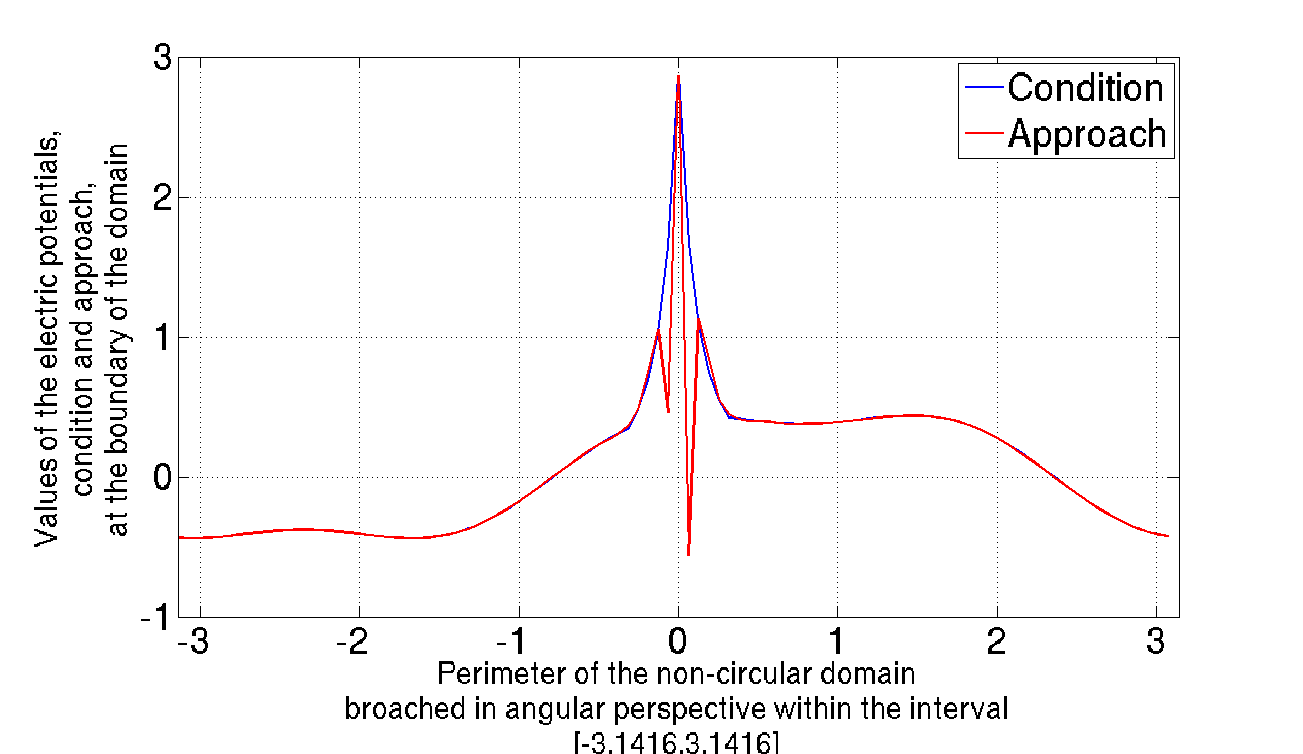}
\caption{Comparison between the boundary condition (\ref{new:lor}) and the numerical approach employing $2N+1=51$ base functions, for the case when $\sigma$ possesses the form plotted in Figure 8.}
\end{figure}

\subsection{A squared figure inside the non-circular domain.}

As a final trial, let us review the case when the conductivity $\sigma$ corresponds to the distribution plotted in Figure 11. This could well be considered a significant example for studying the effectiveness of the posed numerical method. 

\begin{figure}
\centering
\includegraphics[scale=0.3]{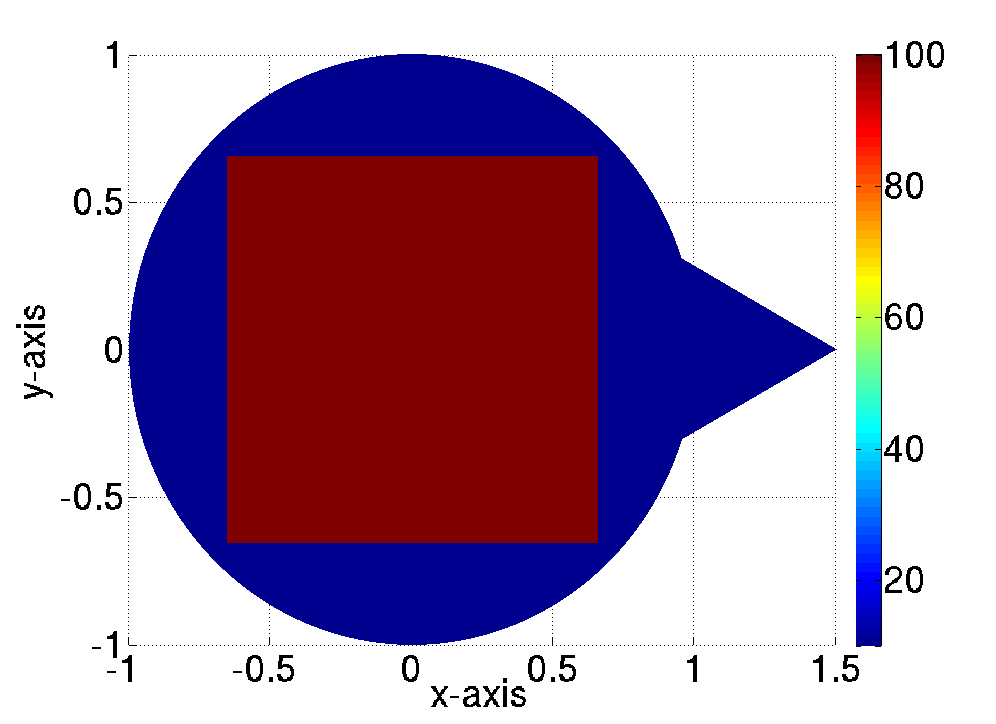}
\caption{Example of a conductivity $\sigma$ with a squared figure inside $\Omega$.}
\end{figure}

We find three points where the derivative of the parametric curve, describing $\Gamma$, is discontinuous, and we have forced one boundary value point to be located at each one of them. Beside, the figure within the domain $\Omega$ possesses four corners. More precisely, it is a perfect square whose characteristics are the same that those employed in the Subsection \ref{square:case}. The first example considered $2N+1=91$ base functions, being its total error: 

\begin{equation}
\mathcal{E}= 0.0010.
\nonumber
\end{equation}

The comparison between the boundary condition and the approached solution is displayed in Figure 12. The last trial was performed utilizing $2N+1=71$ base functions, forcing a boundary value point at $(1.5,0)$.

In the proximity of this point the divergence of the solution is considerable, as displayed in Figure 13, an this is also shown by the increased total error, that for this case was:

\begin{equation}
\mathcal{E}= 2.4867.
\nonumber
\end{equation}

\begin{figure}
\centering
\includegraphics[scale=0.30]{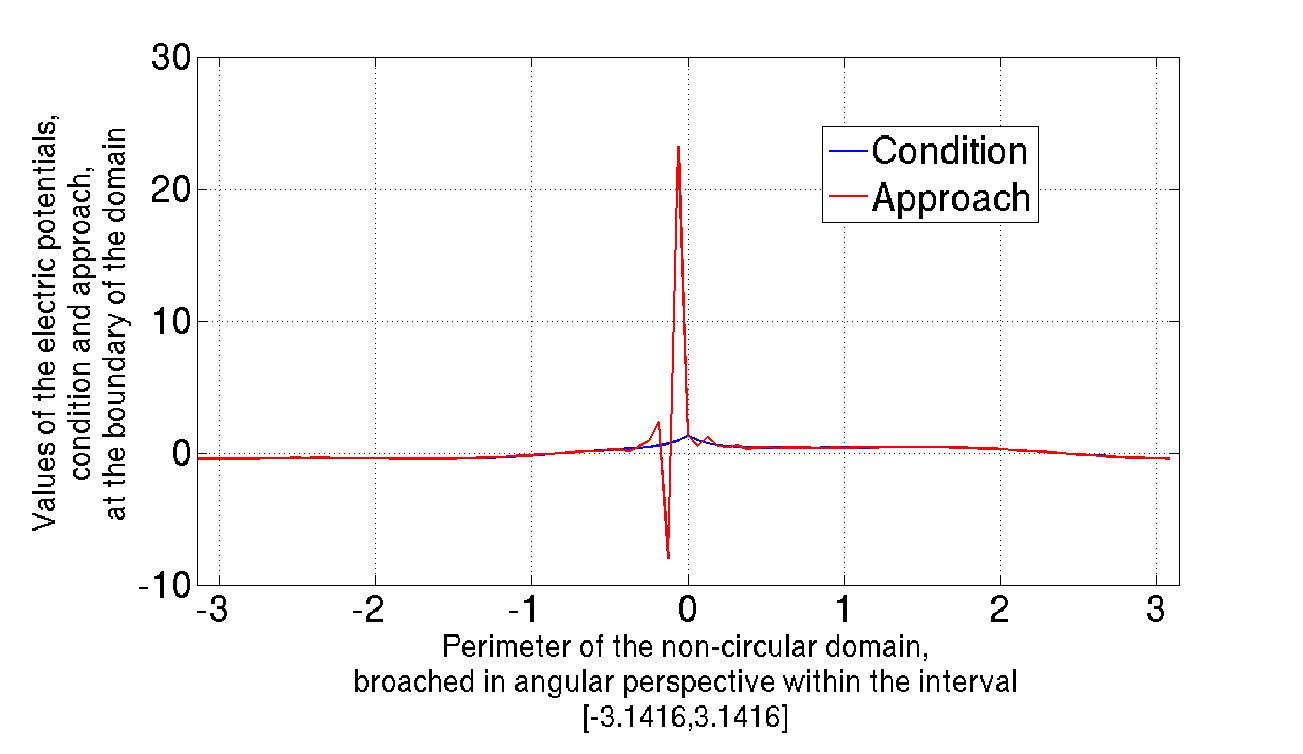}
\caption{Comparison between the boundary condition (\ref{new:lor}) and the numerical approach employing $2N+1=71$ base functions, for the case when $\sigma$ possesses the form plotted in Figure 11.}
\end{figure}

\begin{appointing}
The numerical methods used along this work were fully developed in \textit{GNU C/C++ Compiler}, employing a CPU INTEL\textsuperscript{\textregistered} processor XEON\textsuperscript{\textregistered} E5620$\times 64$B@2.4 GHz, on SLACKWARE 13.37 LINUX operating system. 

The experimental procedures showed that the numerical results can variate when using different platforms based in $32$B and $64$B processor architecture, or compilers between Microsoft\textsuperscript{\textregistered} Windows\textsuperscript{\textregistered} and LINUX operating systems.

If the reader wishes to perform his own numerical trials, please contact the authors to obtain the resource codes.
\end{appointing}

\begin{acknowledgement} The authors would like to acknowledge the support of CONACyT project 106722, and the support of HILMA S.A. de C.V., Mexico. C.M.A. Robles G. acknowledges the support CONACyT project 81599, and La Salle University for the research stay. R.A. Hernandez-Becerril thanks the support of CONACyT. 
\end{acknowledgement}


\begin{thebibliography}{99}
\bibitem{astala}  K. Astala, L. P\"aiv\"arinta, \emph{Calderon's inverse conductivity problem in the
plane}, Annals of Mathematics, Vol. 163, pp. 265-299, 2006.

\bibitem{bers} L. Bers (1953), \emph{Theory of Pseudoanalytic Functions},
IMM, New York University.

\bibitem{bucio} A. Bucio R., R. Castillo-Perez, M.P. Ramirez T. (2011), 
\emph{On the Numerical Construction of Formal Powers and their Application
to the Electrical Impedance Equation}, 8th International Conference on
Electrical Engineering, Computing Science and Automatic Control, IEEE Catalog Number: CFP11827-ART, ISBN:978-1-4577-1013-1, pp. 769-774.

\bibitem{bucio2012} A. Bucio R., R. Castillo-Perez, M.P. Ramirez T., C.M.A. Robles G. (2012), 
\emph{A Simpliﬁed Method for Numerically Solving the Impedance Equation in the Plane}, 9th International Conference on Electrical Engineering, Computing Science and Automatic Control, IEEE Catalog Number:  CFP12827-CDR, ISBN:Catalog Number: CFP 978-1-4673-2168-6, pp. 225-230.

\bibitem{calderon}  A. P. Calderon (1980), \emph{On an inverse boundary value problem}, Seminar on Numerical Analysis and its Applications to Continuum Physics, Sociedade
Brasileira de Matematica, pp. 65-73.

\bibitem{cck} H. M. Campos, R. Castillo-Perez, V. V. Kravchenko (2011), 
\emph{Construction and application of Bergman-type reproducing kernels for boundary and eigenvalue problems in the plane}, Complex Variables and Elliptic Equations, 1-38.

\bibitem{ckr} R. Castillo-Perez., V. Kravchenko, R. Resendiz V. (2011), 
\emph{Solution of boundary value and eigenvalue problems for second order
elliptic operators in the plane using pseudoanalytic formal powers},
Mathematical Methods in the Applied Sciences, Vol. 34, Issue 4.

\bibitem{kpa} V. V. Kravchenko (2009), \emph{Applied Pseudoanalytic Function
Theory}, Series: Frontiers in Mathematics, ISBN: 978-3-0346-0003-3.

\bibitem{kra2005} V. V. Kravchenko (2005), \emph{On the relation of pseudoanalytic function theory to
the two-dimensional stationary Schr\"odinger equation and Taylor series
in formal powers for its solutions}, Journal of Physics A: Mathematical
and General, Vol. 38, No. 18, pp. 3947-3964.

\bibitem{kond}V. A. Kondrat'ev, O. A. Oleinik (1983), \emph{Boundary-value problems for partial differential equations in non-smooth domains}, Russian Mathematical Surveys, IOP.

\bibitem{oct} M. P. Ramirez T. (2010), \emph{On the electrical current
distributions for the generalized Ohm's Law, Applied Mathematics and
Computation}, Elsevier (submitted for publication), available in electronic
at http://arxiv.org

\bibitem{ioprrh} M. P. Ramirez T., R. A. Hernandez-Becerril, M. C. Robles G. (2011), \emph{First characterization of a new method for numerically solving the Dirichlet problem of the two-dimensional Electrical Impedance Equation}, available in electronic
at http://arxiv.org

\bibitem{vekua} I. N. Vekua (1962), \emph{Generalized Analytic Functions,}
International Series of Monographs on Pure and Applied Mathematics, Pergamon
Press.

\bibitem{webster}  J. G. Webster (1990), \emph{Electrical Impedance Tomography}, Adam Hilger
Series on Biomedical Engineering.

\end{thebibliography}
\end{document}